\documentclass[UKenglish]{llncs}
\usepackage{fullpage} 
\usepackage{microtype}

\title{Dichotomy Results for Classified Rank-Maximal Matchings and Popular Matchings}

\author{Meghana Nasre\inst{1} \and Prajakta Nimbhorkar\inst{2} \and Nada Pulath \inst{1}}

\authorrunning{M. Nasre et al.}

\institute{Indian Institute of Technology, Madras, India \and Chennai Mathematical Institute, India}
\usepackage{algorithm}
\usepackage{algorithmic}
\usepackage{amssymb}
\usepackage{amsmath} 
\usepackage{soul}
\usepackage{tikz}
\usetikzlibrary{positioning}
\usetikzlibrary{arrows.meta, positioning, decorations.markings}

\newcommand{\tH}{X}
\newcommand{\tf}{g}

\newcommand{\agt}{applicant }
\newcommand{\agtse}{applicants}

\newcommand{\prgtse}{posts}
\newcommand{\prgte}{post}
\newcommand{\newE}{F}
 \newcommand{\SetU}{U}
\newcommand{\Res}{H_k(f_k)}
\newcommand{\res}{H(f)}
\newcommand{\REM}[1]{}
\newcommand{\NP}{\mbox{{\sf NP}}}
\newcommand{\CRMM}{\mbox{{\sf CRMM}}}
\newcommand{\monEsat}{{\sf monotone 1-in-3 SAT}}

\newcommand{\CHAT}{\mbox{{\sf CHAT}}}
\newcommand{\CPM}{\mbox{{\sf CPM}}}
\newcommand{\LCPM}{\mbox{{\sf LCPM}}}

\begin{document}

\maketitle
\begin{abstract}
In this paper, we consider the problem of computing an optimal matching in a bipartite graph where 
elements of one side of the bipartition specify preferences over the other side, and one or both sides can have capacities and classifications.
The input instance is a bipartite graph $G=(A\cup P,E)$, where $A$ is a set of applicants, $P$ is a set of posts, and
each applicant ranks its neighbors in an order of preference, possibly involving ties. Moreover, each vertex $v\in A\cup P$
has a quota $q(v)$ denoting the maximum number of partners it can have in any allocation of applicants to posts - referred to
as a {\em matching} in this paper.
A classification $\mathcal{C}_u$ for a vertex $u$ is a collection of subsets of neighbors of $u$. Each subset (class) $C\in \mathcal{C}_u$
has an {\em upper quota} denoting the maximum number of vertices from $C$ that can be matched to $u$. The goal is to find a matching
that is {\em optimal} amongst all the {\em feasible matchings}, which are matchings that respect quotas of all the vertices and classes.

We consider two well-studied notions of optimality namely {\em popularity} and {\em rank-maximality}.
The notion of {\em rank-maximality} involves finding a matching in $G$ with maximum number of rank-$1$ edges, subject to that, maximum
number of rank-$2$ edges and so on. 
We present an {$O(|E|^2)$}-time algorithm for finding a feasible rank-maximal matching, when each classification
is a {\em laminar} family. We complement this with an \NP-hardness result when classes are non-laminar even under strict preference lists, and even when only posts have
classifications, and each applicant has a quota of one.
We show an analogous dichotomy result for computing a popular matching  amongst feasible matchings (if one exists) in a bipartite graph with 
posts having capacities and classifications and applicants having a quota of one. 

To solve the classified rank-maximal and popular matchings problems, we present a framework that involves computing max-flows in multiple flow networks.  
We use the fact that, in {\em any} flow network, w.r.t. {\em any} max-flow the vertices can be decomposed
into three disjoint sets and this decomposition is {\em invariant } of the flow. This simple fact turns out to be surprisingly useful
in the design of our combinatorial algorithms. 
We believe that our technique of flow networks 
will find applications
in other capacitated matching problems with preferences.



 \end{abstract}

\section{Introduction}
\label{sec:intro}
The input to our problem is a bipartite graph $G = (A \cup P, E)$ where $A$ is the set of \agtse,
$P$ is the set of \prgtse.
Every vertex $a\in A$ has a preference ordering over  its neighbors in $P$, possibly involving ties, referred to as the 
{\em preference list of $a$}. An edge $(a,p)\in E, a\in A, p\in P$ is said to be a {\em rank-$k$ edge} if $p$ is a $k$-th choice of $a$.
Every vertex $ u \in A \cup P$ specifies a non-zero quota $q(u)$ denoting the maximum number of elements from the other
set it can get matched to. 
Finally, every vertex $u \in A \cup P$ can  specify a classification over its set of neighbors $N(u)$ in $G$. 
A classification $\mathcal{C}_u$ is a family of subsets (referred to as {\em classes} here onwards) of $N(u)$.
Each class $C_u^i\in \mathcal{C}_u$ has an associated quota $q(C_u^i)$ denoting the maximum number of elements
from $C_u^i$ that can be assigned to $u$ in any matching.
\begin{definition} \label{defn:feasible} A matching $M$ is a subset of $E$ and $M(u)$ is the set of all neighbors of $u$ in $M$. An assignment or a matching $M$ in $G$ is said to be {\em feasible} if, for every vertex $u$, the following conditions hold:
\begin{itemize}
\item $|M(u)| \le q(u)$ and
\item for every $C_u^i \in \mathcal{C}_u$, we have $|M(u) \cap C_u^i| \le q(C_u^i)$.
\end{itemize}
\end{definition}
We refer to this setting as the {\em many-to-many setting}, since each vertex can have multiple partners in $M$. A special case is the {\em many-to-one setting}, 
where each applicant can be matched to at most one post, and a post can have multiple applicants matched to it.

Classifications arise naturally in matching problems. While allotting courses to students, a student does not want to be allotted too many courses on closely related topics.
Also, an instructor may not want a course to have too many students from the same department. Another example is allotting tasks to employees, where employees prefer not to
be working on many tasks of similar nature, and for any task, it is wasteful to have too many employees with the same skill-set. These constraints are readily modeled using classifications.

A natural question is to find a feasible matching that is {\em optimal} with respect to the preferences of the \agtse.
In this paper, we consider two well-studied notions of optimality namely {\em rank-maximality} and {\em popularity}. In rank-maximality, the goal is to compute a 
feasible matching in $G$ that has maximum number of rank-1 edges,
subject to this, maximum number of rank-2 edges and so on.
We call such a matching as a {\em Classified Rank-Maximal Matching} (\CRMM). 
The concept of {\em signature}, defined below, is useful to compare two matchings with respect to rank-maximality.
\begin{definition}
The {\em signature} $\sigma_M$ of a matching $M$ is an $r$-tuple $(x_1,\ldots,x_r)$ where $r$ denotes
the largest rank used by an \agt to rank any \prgte. For $1 \le k \le r$,
 $x_k$ denotes the number of rank $k$ edges in $M$.
\end{definition}
\noindent {Let $\sigma_M = (x_1, \ldots , x_r)$ and $\sigma_{M'} = (x'_1, \ldots , x'_r)$.  We say $M \succ M'$ if $x_i = x'_i$ for $1 \leq i < k$ and $x_k > x'_k$, for some $k$.} 
A matching $M$ is said to be {\em rank-maximal} if there does not exist any matching $M'$ in $G$ such that $M' \succ M$.
Thus, our goal is to compute a matching that is rank-maximal among all feasible matchings. We refer to this problem
as the \CRMM\ problem.

In the many-to-one setting, we consider the notion of popularity, which involves comparison of two matchings through the votes
of the applicants.
Given two feasible matchings $M,M'$, an applicant votes
for $M$ if and only if he prefers $M(a)$ over $M'(a)$, and applicants prefer being matched to one of their neighbors over remaining unmatched. 
\begin{definition}
The matching $M$ is {\em more popular than} $M'$ if the number of votes that $M$ gets w.r.t. $M'$ is more than the number of votes that $M'$ gets w.r.t. $M$.  A matching $M$ is 
said to be {\em popular} if there is no matching more popular than $M$. 
\end{definition}
We consider the problem of computing a  popular matching in the presence of classifications, where 
each applicant can be matched to at most one post, and posts have classifications and quotas. Unlike rank-maximal matchings, a popular matching need not exist (see \cite{AIKM07} for a simple instance), since the relation {\em more popular than} is not transitive. 
Our goal therefore is to characterize instances that admit a popular matching and output one if it exists. We call this  the \CPM\ problem.
Note that when a popular matching exists, no majority of applicants can force a migration to another matching; this makes popularity an appealing notion of optimality.

Figure~\ref{fig:example1} shows an example instance where
$A = \{a_1, \ldots, a_5\}$ and $P = \{p_1,\ldots, p_5\}$.
The preferences of the applicants, and the classifications and quotas can be read from the figure. The matching $M~=~\{(a_1, p_4), (a_2, p_1),$ \\$ (a_3, p_3), (a_4, p_5), (a_5, p_2)\}$ is a feasible matching with signature $(3,2)$. The matching $M' = \{(a_1, p_1), (a_2, p_1),$\\$ (a_3, p_3), (a_4, p_5), (a_5, p_2)\}$ has signature $(4, 1)$ but is infeasible because of the classification $C_{p_1}^1$. We will show that the matching is $M$ is both \CRMM\ and \CPM\ in the instance.
\begin{figure}[ht]
\vspace{-0.2in}
\begin{minipage}{0.4\linewidth}
\begin{eqnarray*}
a_1 & : & p_1, p_4 \\
a_2 & : & p_1, p_5\\
a_3 & : & (p_1, p_2, p_3) \\
a_4 & : & p_5, p_1\\
a_5 & : & p_5, p_2
\end{eqnarray*}
{\bf Applicant  Preferences}
\end{minipage}
\begin{minipage}{0.59\linewidth}
\begin{eqnarray*}
\mathcal{C}_{p_1} & = & \{C_{p_1}^1=\{a_1, a_2, a_3\},  \ \ C_{p_1}^2=\{a_4\}\}\\
q(p_1) &=& 2; \ \   q(C_{p_1}^{1})  \ \ =  \ \  q(C_{p_1}^2)  \ \ = \ \ 1 \\
q(p_i) &=& 1 \ \ \ \ {\mbox for}  \ \ \ \ \ i = 2, \ldots, 5 \\
q(a_i) &=& 1   \ \ \ \ {\mbox for}  \ \ \ \ \ i = 1, \ldots, 5 \\
\vspace{-0.008in}
\end{eqnarray*}
\begin{center}{\bf Classifications and Quotas}
\end{center}
\end{minipage}
\caption{Preferences to be read as: $a_1$ treats $p_1$ as rank-$1$ post and $p_2$ as rank-2 post and so on. Applicant $a_3$ treats $p_1, p_2, p_3$ as its rank-1 posts. Although $q(p_1) = 2$, the class $C_{p_1}^1 \in \mathcal{C}_{p_1}$ implies
that  in any feasible matching post $p_1$ can be matched to at most one applicant from $\{a_1, a_2, a_3\}$. }
\label{fig:example1}
\end{figure}

Matchings in the presence of preferences and  classifications have been studied in the setting 
where both sides of the bipartition have preferences over the other side. {\em Stability}~\cite{GS62} is a widely
accepted notion of optimality in this setting. Huang \cite{Huang10} 
considered the stable matching problem in the many-to-one case, where one side of the bipartition has classifications. This was later extended to the many-to-many setting where both
sides have classifications \cite{FK12}.
We remark that the setting in \cite{Huang10} and \cite{FK12} involves both upper and lower quotas on vertices and classes, whereas
our setting has only upper quotas. 
However, this problem has not been studied in the case where only one side of the bipartition expresses preferences. 

%

In the stable matching case, existence of a stable matching respecting the classifications can be determined in polynomial-time
if the classes specified by each vertex form a {\em laminar} family
\cite{Huang10,FK12}, and otherwise the problem is \NP-complete \cite{Huang10}.
In our setting, the preferences
being only on one side and the optimality criteria being rank-maximality or popularity
are very different from the stable matching setting. Yet we show similar results as those of \cite{Huang10} and \cite{FK12}.
 A family $\mathcal{F}$ of subsets of a set $S$ 
is said to be {\em laminar} if, for every pair of sets $X,Y\in \mathcal{F}$, either $X\subseteq Y$ or $Y\subseteq X$ or $X\cap Y=\emptyset$.
Laminar classifications are natural in settings like student allocation to schools where schools may want at most a certain number of students
from a particular region, district, state, country and so on.
Laminar classification includes the special case of {\em partition}, where the classes are required to be disjoint.
This is a very natural classification arising in many real-world applications.

\subsection{Our Contribution}
We show the following new results in this paper. Let $G = (A \cup P, E)$ denote an instance of the \CRMM\ problem or the \CPM\ problem.
\begin{theorem}\label{thm:poly}
There is an $O(|E|^2)$-time algorithm for the \CRMM\ problem when the classification for every vertex is a laminar family.
\end{theorem}
We also show the above result for the \CPM\ problem in the many-to-one setting.
\begin{theorem}\label{thm:pop-poly}
There is an $O(|A||E|)$-time algorithm for the \CPM\ problem when the classification for every post is a laminar family.
\end{theorem}
We complement the above results with a matching hardness result:
\begin{theorem}\label{thm:crmm-hard}
The \CRMM\ and \CPM\ problems are \NP-hard when the classes are non-laminar even when all the preferences are strict, 
and classifications exist on only one side of the bipartition.
\end{theorem}
The hardness holds even when the intersection of the classes in a family is at most one, and the preference lists have length at most $2$.
Even when there are no ranks on edges,
the problem of simply finding a maximum cardinality matching respecting the classifications is \NP-hard if the classes are non-laminar.
\begin{theorem}\label{thm:max-hard}
The problem of finding a maximum cardinality matching is \NP-hard in the presence of non-laminar classifications.
\end{theorem}

\noindent {\bf Related work:}
Irving introduced the rank-maximal matchings  problem as ``greedy matchings" in \cite{greedyI03} for the one-to-one case of strict preferences.
Irving~et~al.~\cite{IKMMP06} generalized the same to preference lists with ties allowed and this was further generalized by Paluch~\cite{Paluch13} for 
the many-to-many setting.
Abraham~et~al.~\cite{AIKM07} initiated the study of Popular Matchings problem in the one-to-one setting and subsequently there have been several results~\cite{MS06,M08,HKMN08} on generalization of this model.
In all the above results where the model is without classifications, the algorithms for computing a rank-maximal matching \cite{IKMMP06,Paluch13} and for computing popular matching in \cite{AIKM07,MS06,M08,HKMN08} have the following template: 
The algorithms are iterative, where iteration $k$ involves the instance restricted to edges of rank at most $k$. All of the above results make crucial use
of the 
well-known Dulmage-Mendelsohn decomposition w.r.t. maximum matchings in bipartite graphs.
The main use of the decomposition theorem in all the literature mentioned above is to identify edges that can not belong to any optimal matching.
Such edges are deleted in each iteration, resulting in a {\em reduced graph}, such that every maximum matching in the reduced graph is an optimal matching in the given instance.

\noindent{\bf Our technique: }
In our setting, we have quotas as well as classifications. Hence 
a feasible matching need not be a maximum matching even in the reduced graph and therefore the Dulmage-Mendelsohn decomposition~\cite{DulmageM58} can not be
used as in \cite{IKMMP06,Paluch13}. 
To solve the \CRMM\ and \CPM\ problems, we present a framework that involves computing max-flows in multiple flow networks.  While the use
of flow network is a natural choice for laminar classifications, it still leaves us with the challenge of identifying the set of
unnecessary edges. We address this by using the fact that, in {\em any} flow network, w.r.t. {\em any} max-flow the vertices can be decomposed 
into three disjoint sets and this decomposition is {\em invariant } of the flow. This simple fact turns out to be surprisingly useful and allows
us to use the forward and reverse edges of a min-cut to identify unnecessary edges. We believe that our technique of flow networks provides a unified
framework for capacitated rank-maximal matchings~\cite{Paluch13} and capacitated house allocation problem~\cite{MS06} and will find further applications
in capacitated matching problems with preferences.
We finally note that the \CRMM\ problem can also be solved using min-cost flows with slightly higher time complexity, but that 
approach involves using exponential weights.
Our algorithm is simple, combinatorial and uses only elementary flow computations and also extends to the \CPM\ problem.




\noindent {\em Organization of the paper:} 
In Section~\ref{sec:algo} we describe our flow network for the laminar \CRMM\ problem and prove  properties of the network.
In Section~\ref{sec:pseudocode} we present our algorithm  and prove its correctness. 
We present the detailed algorithmic results for the \CPM\ problem in Section~\ref{sec:pop}.
In Section~\ref{sec:hardness} we give the hardness for the non-laminar \CRMM\ problem.

\section{Laminar \CRMM}
\label{sec:algo}
In this section, we present the construction of our flow-networks used by the polynomial-time algorithm for the \CRMM\ problem when the classes of each vertex
form a laminar family. Recall that the given instance is a bipartite graph $G=(A\cup P,E)$, along with a preference list
for each $a\in A$, and a laminar classification $\mathcal{C}_u$ for each $u\in A\cup P$.
The algorithm starts by constructing a flow network $H_0$
using the classifications. 
Our algorithm then works in iterations. 
In the $k$-th iteration, we add rank-$k$ edges from $G$
to the flow network $H_{k-1}$ to get a new flow network $H_k$.
We find a max-flow $f_k$ in $H_k$ 
and then identify and delete {\em unnecessary} edges. Throughout the course of the algorithm, we maintain the following invariant: At the end of each iteration $k$,
there is a matching $M_k$ in $G$ corresponding to 
$H_k$, such that the signature of $M_k$ is $(s_1,s_2,\ldots, s_k)$ where $(s_1,\ldots,s_r)$
is the signature of a feasible rank-maximal matching in $G$. 

Algorithm~\ref{algo:main} in Section~\ref{sec:pseudocode} gives a formal pseudocode. We first show the construction and properties of the flow network,
and then give a detailed description and correctness proof for Algorithm \ref{algo:main}.

\subsection{Construction of flow network}
\label{sec:flownw}
We describe the construction of the flow network $H_0$ corresponding to the input bipartite graph $G$ with classifications. 
As mentioned above, the $k$th iteration of the algorithm uses the flow network $H_k$. 
The vertex set of the flow network is the same for each $k$, and hence we refer to the initial flow
network as $H_0 =  (V, \newE_0)$.
We apply the following pre-processing step for every vertex in $G$:

For every $u \in A \cup P$ with classification $C_u$, we add the following classes to $\mathcal{C}_u$.
\begin{itemize}
    \item $C_u^*$: We include a class $C_u^*=N(u)$ into $\mathcal{C}_u$ with capacity $q(C_u^*)$ = $q(u)$.
    \item $C_u^w$: For every $w \in N(u)$ and $u \in A \cup P$, we add a class $C_u^w$ to $\mathcal{C}_u$ with capacity $q(C_u^w)=1$.
\end{itemize}
It is easy to see that this does not
change the set of feasible matchings. In the rest of the paper, we refer to this modified instance as our instance $G$.
\begin{definition}[Classification tree:]
\label{def:classtree}
 Let every vertex $u \in A \cup P$ have a laminar family of classes $\mathcal{C}_u$. Then, the classes in $\mathcal{C}_u$ can be represented as a tree called the {\it classification tree} $\mathcal{T}_u$ with $C_u^*$ being the root of $\mathcal{T}_u$. For two classes $C_u^1,  C_u^2 \in \mathcal{C}_u$, the class $C_u^1$ is a parent of $C_u^2$ in $\mathcal{T}_u$ iff $C_u^1$ is the smallest class in $\mathcal{C}_u$ containing $C_u^2$. Thus for every $w \in N(u)$, the corresponding singleton class $C_u^w$ is a leaf of $\mathcal{T}_u$.
\end{definition}

Through out the paper, we refer to the vertices $V$ of $H_0$ as ``nodes''.
The network $H_0$ has nodes corresponding to every element of $\mathcal{T}_u$ for each
$u \in A \cup P$. In addition to this, there is a source $s$ and a sink $t$. The edges of $H_0$ include an edge from $s$ to the root of $\mathcal{T}_a$ for each $a \in A$, and an  edge from the root of $\mathcal{T}_p$ to $t$, for each $p \in P$.
Each edge of $\mathcal{T}_a$, for each $a \in A$, is directed from parent to child whereas each edge of $\mathcal{T}_p$, $p\in P$ is directed from child to parent in $H_0$. 
This is summarized below:
\begin{eqnarray*}
V = \{s, t\} \cup \{C_u^i \mid C_u^i \in \mathcal{T}_u \mbox { and } u \in A \cup P\}
\end{eqnarray*}
The set of all edges of $H_0$ represented by $\newE_0$ and their capacities are as follows:
\begin{itemize}
\item For every $a \in A$, $\newE_0$ contains an edge $(s, C_a^*)$ with capacity $q(C_a^*)$.
\item For every $p \in P$, $\newE_0$ contains an edge $(C_p^*, t)$ with capacity $q(C_p^*)$.
\item For $a \in A$ and edge $(C_a^1,C_a^2) \in \mathcal{T}_a$ such that $C_a^1$ is the parent of $C_a^2$, $\newE_0$ contains
an edge $(C_a^1,C_a^2)$ with capacity $q(C_a^2)$.
\item For $p \in P$ and edge $(C_p^1,C_p^2) \in \mathcal{T}_p$ such that $C_p^2$ is the parent of $C_p^1$, $\newE_0$ contains
an edge $(C_a^1,C_a^2)$ with capacity $q(C_a^1)$.
\end{itemize}
\tikzset{middlearrow/.style={
        decoration={markings,
            mark= at position 0.5 with {\arrow{#1};} ,
        },
        postaction={decorate}
    }
}
\newcommand{\boundellipse}[3]
{(#1) ellipse (#2 and #3)
}
\begin{figure}[t] 
\begin{minipage}{0.45 \textwidth}
\hspace{-1.3cm}
\scalebox{0.7}{\begin{tikzpicture}[
roundnode/.style={circle, draw=black!100,  inner sep=0pt, minimum size=8pt}, unode/.style={circle, draw=black!100,  inner sep=0pt, minimum size=8pt}, tnode/.style={circle, draw=black!100,  inner sep=0pt, minimum size=8pt},
farrow/.style={-Latex}
]

\node[roundnode]      (a4p1)            [label=above:$C^{p_1}_{a_4}$] {};
\node[tnode]        (a3p2)       [below=7.5mm of a4p1, label=above:$C^{p_2}_{a_3}$] {};
\node[roundnode]        (a5p2)       [below=7.5mm of a3p2, label=above:$C^{p_2}_{a_5}$] {};
\node[tnode]        (a3p3)       [below=7.5mm of a5p2, label=above:$C^{p_3}_{a_3}$] {};
\node[roundnode]        (a4p5)       [below=7.5mm of a3p3, label=above:$C^{p_5}_{a_4}$] {};
\node[roundnode]        (a5p5)       [below=7.5mm of a4p5, label=above:$C^{p_5}_{a_5}$] {};
\node[unode]        (a3p1)       [above=7.5mm of a4p1, label=above:$C^{p_1}_{a_3}$] {};
\node[roundnode]        (a2p1)       [above=15mm of a3p1, label=above:$C^{p_1}_{a_2}$] {};
\node[roundnode]        (a2p5)       [below=6.9mm of a2p1, label=above:$C^{p_5}_{a_2}$] {};
\node[roundnode]        (a1p1)       [above=7.5mm of a2p1, label=above:$C^{p_1}_{a_1}$] {};
\node[roundnode]        (a1p4)       [above=7.5mm of a1p1, label=above:$C^{p_4}_{a_1}$] {};
\node[tnode]        (p4a1)       [right=15mm of a1p4, label=above:$C^{a_1}_{p_4}$] {};
\node[roundnode]        (p1a1)       [right=15mm of a1p1, label=above:$C^{a_1}_{p_1}$] {};
\node[roundnode]        (p1a2)       [right=15mm of a2p1, label=above:$C^{a_2}_{p_1}$] {};
\node[unode]        (p5a2)       [right=15mm of a2p5, label=above:$C^{a_2}_{p_5}$] {};
\node[unode]        (p1a3)       [right=15mm of a3p1, label=above:$C^{a_3}_{p_1}$] {};
\node[tnode]        (p1a4)       [right=15mm of a4p1, label=above:$C^{a_4}_{p_1}$] {};
\node[tnode]        (p2a3)       [right=15mm of a3p2, label=above:$C^{a_3}_{p_2}$] {};
\node[tnode]        (p2a5)       [right=15mm of a5p2, label=above:$C^{a_5}_{p_2}$] {};
\node[tnode]        (p3a3)       [right=15mm of a3p3, label=above:$C^{a_3}_{p_3}$] {};
\node[roundnode]        (p5a4)       [right=15mm of a4p5, label=above:$C^{a_4}_{p_5}$] {};
\node[roundnode]        (p5a5)       [right=15mm of a5p5, label=above:$C^{a_5}_{p_5}$] {};

\node[roundnode]        (a1)       [left=15mm of a1p1, label=above:$C^{*}_{a_1}$] {};
\node[roundnode]        (a2)       [left=15mm of a2p5, label=above:$C^{*}_{a_2}$] {};
\node[tnode]        (a3)       [left=15mm of a4p1, label=above:$C^{*}_{a_3}$] {};
\node[roundnode]        (a4)       [left=15mm of a5p2, label=below:$C^{*}_{a_4}$] {};
\node[roundnode]        (a5)       [left=15mm of a4p5, label=below:$C^{*}_{a_5}$] {};
\node[roundnode]        (s)       [left=15mm of a3, label=left:$s$] {};

\node[tnode]        (p4)       [right=34.5mm of p1a1, label=above:$C^{*}_{p_4}$] {};
\node[roundnode]        (Cp11)       [right=15mm of p1a2, label=above:$C^{1}_{p_1}$] {};
\node[tnode]        (Cp12)       [right=15mm of p1a4, label=above:$C^{2}_{p_1}$] {};
\node[tnode]        (p1)       [right=34.5mm of p1a3, label=above:$C^{*}_{p_1}$] {};
\node[tnode]        (p3)       [right=34.5mm of p3a3, label=below:$C^{*}_{p_3}$] {};
\node[tnode]        (p2)       [above=10.5mm of p3, label=above:$C^{*}_{p_2}$] {};
\node[roundnode]        (p5)       [below=12.5mm of p3, label=below:$C^{*}_{p_5}$] {};
\node[tnode]        (t)       [right=34.5mm of Cp12, label=right:$t$] {};

\draw[black,thick] (0,0.2) ellipse (0.68 cm and 5.8 cm);
\draw[black,thick] (1.8,0.2) ellipse (0.68 cm and 5.8 cm);
\draw (0,-5.65) coordinate[label={below:$L$}] (L);
\draw (2.1,-5.65) coordinate[label={below:$R$}] (R);
\begin{scope}[very thick]
\draw[middlearrow={>}] (s.east) -- (a1.south); 
\draw[middlearrow={>}] (s.east) -- (a2.south);
\draw[middlearrow={>}] (s.east) -- (a3.west); 
\draw[middlearrow={>}] (s.east) -- (a4.north);
\draw[middlearrow={>}] (s.east) -- (a5.north);
\draw[middlearrow={<}] (t.west) -- (p2.east);
\draw[middlearrow={<}] (t.east) -- (p3.north); 
\draw[middlearrow={<}] (t.north) -- (p4.east);
\draw[middlearrow={<}] (t.south) -- (p5.north);
\draw[middlearrow={>}] (a1.east) -- (a1p1.west); 
\draw[middlearrow={>}] (a1.east) -- (a1p4.west); 
\draw[middlearrow={>}] (a2.east) -- (a2p1.west); 
\draw[middlearrow={>}] (a3.east) -- (a3p1.west); 
\draw[middlearrow={>}] (a3.east) -- (a3p2.west); 
\draw[middlearrow={>}] (a3.east) -- (a3p3.west); 
\draw[middlearrow={>}] (a4.north) -- (a4p1.west); 
\draw[middlearrow={>}] (a5.east) -- (a5p5.west); 
\draw[middlearrow={>}] (p4a1.east) -- (p4.west); 
\draw[middlearrow={>}] (p1a1.east) -- (Cp11.west); 
\draw[middlearrow={>}] (Cp11.east) -- (p1.west);
\draw[middlearrow={>}] (p1a2.east) -- (Cp11.west); 
\draw[middlearrow={>}] (p1a3.east) -- (Cp11.west); 
\draw[middlearrow={>}] (p1a4.east) -- (Cp12.west); 
\draw[middlearrow={>}] (Cp12.east) -- (p1.west); 
\draw[middlearrow={>}] (p2a5.east) -- (p2.west); 
\draw[middlearrow={>}] (p3a3.east) -- (p3.west); 
\draw[middlearrow={>}] (p5a4.east) -- (p5.west); 
\draw[middlearrow={>}] (p5a5.east) -- (p5.west); 
\draw[middlearrow={>}]  (p1.east) to node [auto] {$2$} (t.west);
\draw[middlearrow={>}] (a2.east) -- (a2p5.west);

\end{scope}
\begin{scope}[decoration={
    markings, mark=at position 0.45 with {\arrow{>}}}
    ] 
\draw[very thick,postaction={decorate}] (a4.west) -- (a4p5.east); 
\draw[very thick,postaction={decorate}] (a5.east) -- (a5p2.west); 
\draw[very thick,postaction={decorate}]  (p5a2.east)-- (p5.west);
\draw[very thick,postaction={decorate}] (p2a3.east) -- (p2.west); 
\end{scope}

\begin{scope}[very thick, dashed, decoration={
    markings,
    mark=at position 0.5 with {\arrow{>}}}
    ] 

\end{scope}

\end{tikzpicture}
}
\caption{ 
The flow network $H_0$ corresponding to instance in Figure~\ref{fig:example1}. All edges except $(C_{p_1}^*, t)$ have unit capacity. The capacity of $(C_{p_1}^*, t)$ equals $q(p_1) = 2$.}
\label{new1}

\end{minipage}
\begin{minipage}{0.1 \textwidth}
\hspace{0.5cm}
\end{minipage}
\begin{minipage}{0.45 \textwidth}

\vspace{1.2cm}
\scalebox{0.7}{\begin{tikzpicture}[
roundnode/.style={circle, draw=black!100, inner sep=0pt, minimum size=8pt}, unode/.style={circle, draw=red!100, fill = red, inner sep=0pt, minimum size=8pt}, tnode/.style={circle, draw=black!100, fill = black, inner sep=0pt, minimum size=8pt},
farrow/.style={-Latex}
]

\node[roundnode]      (a4p1)            [label=above:$C^{p_1}_{a_4}$] {};
\node[tnode]        (a3p2)       [below=7.5mm of a4p1, label=above:$C^{p_2}_{a_3}$] {};
\node[roundnode]        (a5p2)       [below=7.5mm of a3p2, label=above:$C^{p_2}_{a_5}$] {};
\node[tnode]        (a3p3)       [below=7.5mm of a5p2, label=above:$C^{p_3}_{a_3}$] {};
\node[roundnode]        (a4p5)       [below=7.5mm of a3p3, label=above:$C^{p_5}_{a_4}$] {};
\node[roundnode]        (a5p5)       [below=7.5mm of a4p5, label=above:$C^{p_5}_{a_5}$] {};
\node[unode]        (a3p1)       [above=7.5mm of a4p1, label=above:$C^{p_1}_{a_3}$] {};
\node[roundnode]        (a2p1)       [above=15mm of a3p1, label=above:$C^{p_1}_{a_2}$] {};
\node[roundnode]        (a2p5)       [below=6.9mm of a2p1, label=above:$C^{p_5}_{a_2}$] {};
\node[roundnode]        (a1p1)       [above=7.5mm of a2p1, label=above:$C^{p_1}_{a_1}$] {};
\node[roundnode]        (a1p4)       [above=7.5mm of a1p1, label=above:$C^{p_4}_{a_1}$] {};
\node[tnode]        (p4a1)       [right=15mm of a1p4, label=above:$C^{a_1}_{p_4}$] {};
\node[roundnode]        (p1a1)       [right=15mm of a1p1, label=above:$C^{a_1}_{p_1}$] {};
\node[roundnode]        (p1a2)       [right=15mm of a2p1, label=above:$C^{a_2}_{p_1}$] {};
\node[unode]        (p5a2)       [right=15mm of a2p5, label=above:$C^{a_2}_{p_5}$] {};
\node[unode]        (p1a3)       [right=15mm of a3p1, label=above:$C^{a_3}_{p_1}$] {};
\node[tnode]        (p1a4)       [right=15mm of a4p1, label=above:$C^{a_4}_{p_1}$] {};
\node[tnode]        (p2a3)       [right=15mm of a3p2, label=above:$C^{a_3}_{p_2}$] {};
\node[tnode]        (p2a5)       [right=15mm of a5p2, label=above:$C^{a_5}_{p_2}$] {};
\node[tnode]        (p3a3)       [right=15mm of a3p3, label=above:$C^{a_3}_{p_3}$] {};
\node[roundnode]        (p5a4)       [right=15mm of a4p5, label=above:$C^{a_4}_{p_5}$] {};
\node[roundnode]        (p5a5)       [right=15mm of a5p5, label=above:$C^{a_5}_{p_5}$] {};

\node[roundnode]        (a1)       [left=15mm of a1p1, label=above:$C^{*}_{a_1}$] {};
\node[roundnode]        (a2)       [left=15mm of a2p5, label=above:$C^{*}_{a_2}$] {};
\node[tnode]        (a3)       [left=15mm of a4p1, label=above:$C^{*}_{a_3}$] {};
\node[roundnode]        (a4)       [left=15mm of a5p2, label=below:$C^{*}_{a_4}$] {};
\node[roundnode]        (a5)       [left=15mm of a4p5, label=below:$C^{*}_{a_5}$] {};
\node[roundnode]        (s)       [left=15mm of a3, label=left:$s$] {};

\node[tnode]        (p4)       [right=34mm of p1a1, label=above:$C^{*}_{p_4}$] {};
\node[roundnode]        (Cp11)       [right=15mm of p1a2, label=above:$C^{1}_{p_1}$] {};
\node[tnode]        (Cp12)       [right=15mm of p1a4, label=above:$C^{2}_{p_1}$] {};
\node[tnode]        (p1)       [right=34.5mm of p1a3, label=above:$C^{*}_{p_1}$] {};
\node[tnode]        (p3)       [right=34.5mm of p3a3, label=below:$C^{*}_{p_3}$] {};
\node[tnode]        (p2)       [above=10.5mm of p3, label=above:$C^{*}_{p_2}$] {};
\node[roundnode]        (p5)       [below=12.5mm of p3, label=below:$C^{*}_{p_5}$] {};
\node[tnode]        (t)       [right=34.5mm of Cp12, label=right:$t$] {};

\draw[black,thick] (0,0.2) ellipse (0.68 cm and 5.8 cm);
\draw[black,thick] (1.8,0.2) ellipse (0.68 cm and 5.8 cm);
\draw (0,-5.65) coordinate[label={below:$L$}] (L);
\draw (2.1,-5.65) coordinate[label={below:$R$}] (R);
\begin{scope}[very thick]
\draw[middlearrow={<}] (s.east) -- (a1.south); 
\draw[middlearrow={>}] (s.east) -- (a2.south);
\draw[middlearrow={<}] (s.east) -- (a3.west); 
\draw[middlearrow={<}] (s.east) -- (a4.north);
\draw[middlearrow={>}] (s.east) -- (a5.north);
\draw[middlearrow={>}] (t.west) -- (p2.east);
\draw[middlearrow={<}] (t.east) -- (p3.north); 
\draw[middlearrow={<}] (t.north) -- (p4.east);
\draw[middlearrow={>}] (t.south) -- (p5.north);
\draw[middlearrow={<}] (a1.east) -- (a1p1.west); 
\draw[middlearrow={>}] (a1.east) -- (a1p4.west); 
\draw[middlearrow={>}] (a2.east) -- (a2p1.west); 
\draw[middlearrow={>}] (a3.east) -- (a3p1.west); 
\draw[middlearrow={<}] (a3.east) -- (a3p2.west); 
\draw[middlearrow={>}] (a3.east) -- (a3p3.west); 
\draw[middlearrow={>}] (a4.north) -- (a4p1.west); 
\draw[middlearrow={>}] (a5.east) -- (a5p5.west); 
\draw[middlearrow={<}] (a1p1.east) -- (p1a1.west); 
\draw[middlearrow={>}] (a2p1.east) -- (p1a2.west); 
\draw[middlearrow={>}] (a3p1.east) -- (p1a3.west); 
\draw[middlearrow={<}] (a3p2.east) -- (p2a3.west); 
\draw[middlearrow={>}] (a3p3.east) -- (p3a3.west); 
\draw[middlearrow={<}] (a4p5.east) -- (p5a4.west); 
\draw[middlearrow={>}] (a5p5.east) -- (p5a5.west); 
\draw[middlearrow={>}] (p4a1.east) -- (p4.west); 
\draw[middlearrow={<}] (p1a1.east) -- (Cp11.west); 
\draw[middlearrow={<}, color=gray!60] (Cp11.east) -- (p1.west);
\draw[middlearrow={>}] (p1a2.east) -- (Cp11.west); 
\draw[middlearrow={>}, color=gray!60] (p1a3.east) -- (Cp11.west); 
\draw[middlearrow={>}] (p1a4.east) -- (Cp12.west); 
\draw[middlearrow={>}] (Cp12.east) -- (p1.west); 

\draw[middlearrow={>}] (p2a5.east) -- (p2.west); 
\draw[middlearrow={>}] (p3a3.east) -- (p3.west); 
\draw[middlearrow={<}] (p5a4.east) -- (p5.west); 
\draw[middlearrow={>}] (p5a5.east) -- (p5.west); 
\draw[bend right,middlearrow={>}]  (p1.east) to node [auto] {} (t.west);
\draw[bend left,middlearrow={<}]  (p1.east) to node [auto] {} (t.west);
\draw[middlearrow={>}] (a2.east) -- (a2p5.west);

\end{scope}
\begin{scope}[decoration={
    markings, mark=at position 0.45 with {\arrow{>}}}
    ] 
\draw[very thick,postaction={decorate}] (a4p5.west) -- (a4.east); 
\draw[very thick,postaction={decorate}] (a5.east) -- (a5p2.west); 
\draw[very thick,postaction={decorate}, color=gray!60]  (p5a2.east)-- (p5.west);
\draw[very thick,postaction={decorate}] (p2.east) -- (p2a3.west); 
\end{scope}

\begin{scope}[very thick, dashed, decoration={
    markings,
    mark=at position 0.5 with {\arrow{>}}}
    ] 
\draw[middlearrow={>}] (a4p1.east) -- (p1a4.west); 
\draw[middlearrow={>}] (a5p2.east) -- (p2a5.west); 
\draw[middlearrow={>}] (a1p4.east) -- (p4a1.west); 

\end{scope}

\end{tikzpicture}}
\caption{Thick edges and gray edges form the network $H_1(f_1)$. Thick edges alone form the network $H'_1$. Thick and dashed edges together form
the network $H_2$.
The thick and dashed edges between $L$ and $R$ represent rank-$1$ and rank-$2$ edges respectively. The white, black, and red 
nodes represent $S_1$, $T_1$ and $U_1$ respectively.
}
\label{new2}
\end{minipage}
\end{figure}


We collectively refer to the set of leaves of $\mathcal{T}_a$ for all $a\in A$ as $L$ and similarly, the set of leaves
of $\mathcal{T}_p$ for all $p\in P$ as $R$. Thus
\begin{eqnarray*}
L = \{ C_a^p \mid a \in A \mbox { and } p \in N(a) \}; & &R  =  \{ C_p^a \ \ \ | \ \ \  p \in P \mbox { and } a \in N(p) \}
\end{eqnarray*}
Figure~\ref{new1} shows the flow network corresponding to the example in Figure~\ref{fig:example1}.
The nodes in $L$ (respectively $R$) (shown in the two ellipses in the figure) have a unique
predecessor (successor) in $H_0$. Moreover, $H_0$ can be seen as a disjoint union of two trees, one rooted at $s$ and
another at $t$, the edges of the former being directed from parent to child and those of the latter from child to parent.
We call the two trees as {\em applicant-tree} and {\em post-tree} respectively.

\subsubsection {Decomposition of vertices}
\label{sec:decompose}
In this section, we present a decomposition of the vertices of the flow network w.r.t. a max-flow. As evident,
the graph $H_0$ admits no path from $s$ to $t$, hence has a zero max-flow. Our algorithm in Section~\ref{sec:pseudocode} iteratively adds edges to $H_0$. 
In an iteration $k$, the flow network $H_k$ contains  unit capacity edges of the form $(C_a^p, C_p^a)$ between the sets $L$ and $R$
such that $p \in N(a)$ and the edge $(a, p)$ has rank at most $k$. 
Let $H$ be any such flow network constructed by our algorithm in some iteration and let $f$ be a max-flow in $H$.
We give a decomposition of vertices of $H$ w.r.t. the max-flow $f$. 
We prove in Section~\ref{sec:flow-prop} that the decomposition is {\em invariant } of the max-flow. The decomposition
of the vertices allows us to delete certain edges in $H$ that ensures that signature of the matching $M$ corresponding to $H$ is preserved in the future iterations.
For a flow network $H$ and a max-flow $f$ in $H$, let $\res$ denote the residual network.
We define the sets $S_f, T_f, \SetU_f$ as follows. Since $f$ is a max-flow, it is immediate that the sets
partition the vertex set $V$.
\begin{eqnarray*}
S_f &=& \{ v \mid v \in V \mbox{ and } v \mbox { is reachable from } s \mbox{ in } \res\} \\ 
T_f &=& \{ v \mid v \in V \mbox{ and } v \mbox { can reach } t \mbox{ in } \res\}\\
\SetU_f &=& \{ v \mid v \in V \mbox{ and } v \notin S_f \cup T_f \}
\end{eqnarray*}

\subsection{Properties of the flow network}\label{sec:flow-prop}
We state properties of the flow network which are essential to prove the correctness of Algorithm~\ref{algo:main}.
Lemma~\ref{lem:diffcuts} and Lemma~\ref{lem:for-rev-flow} below are known from theory of network flows (See e.g. \cite{FF}). 
Lemma~\ref{lem:invariant} shows the invariance of the sets $S_f, T_f, U_f$.
We remark that the properties in Lemma~\ref{lem:diffcuts}, \ref{lem:for-rev-flow}, and \ref{lem:invariant} hold for any flow network $H$.

\begin{lemma} \label{lem:diffcuts}
Let $f$ be a max-flow in a flow network $H=(V,E)$ and $S_f, T_f$, and $\SetU_f$ be as defined above using the residual network $\res$.
$(S_f, T_f \cup U_f)$ is a min-$s$-$t$-cut of $H$.
\end{lemma}

\begin{lemma}
\label{lem:for-rev-flow}
Let $H$ be any flow network and $f$ be a max-flow in $H$. Let $(X, Y)$ be any min-$s$-$t$-cut of $H$.
Then the following hold:
\begin{itemize}
\item For any edge $(a, b)\in E$ such that $a\in X, b\in Y$, we have $f(a, b) = c(a, b)$.
\item For any edge $(b, a)\in E$ such that $a\in X, b\in Y$, we have $f(a, b) = 0$.
\end{itemize}
\end{lemma}

\begin{lemma} 
\label{lem:invariant}
The sets $S_f, T_f$ and $U_f$ are invariant of the max-flow $f$ in $H$.
\end{lemma}

\begin{proof}
Let $f$ and $f'$ be two max-flows in $H$.
Let $S_f, T_f, U_f$ be the sets w.r.t. $f$ and $S_{f'}, T_{f'}, U_{f'}$ be the sets w.r.t. $f'$.
We consider the following two cases.
\begin{itemize}
\item We show that, for any node $x \in H$,  $x \in S_f \iff x \in S_{f'}$.
We prove one direction i.e.
$x \in S_f \implies x \in S_{f'}$. The other direction follows by symmetry.
For the sake of contradiction, assume that there exists an $x \in S_f$ such that 
$x \in  T_{f'} \cup U_{f'}$. 
Furthermore among all nodes in $S_f\setminus S_{f'}$, let $x$ be the one whose shortest path distance from $s$
in the residual network $\res$ is as small as possible. 

Let $y$ be the parent of $x$ in the BFS tree rooted at $s$ in $\res$.
By the choice of $x$, it is clear that $y \in S_{f'}$. (We remark that $y$ could be the node $s$ itself.)  Note that
 the edge $(y, x)$ belongs to $\res$.  
 Therefore either $(y,x)\in \res$ or $(x,y)\in\res$. 
If $(y,x)\in \res$, then $(y,x)$ is a forward edge of the min-$s$-$t$-cut $(S_{f'}, T_{f'} \cup U_{f'})$, and hence must be saturated by $f$ as well. Thus, by Lemma~\ref{lem:for-rev-flow}, $f'(y, x) = f(y, x) = c(y, x)$. However, this contradicts
the fact that $(y, x) \in \res$.

If $(x, y) \in H$ then $(x,y)$ is a reverse edge of the min-$s$-$t$-cut $(S_{f'}, T_{f'} \cup U_{f'})$.
Hence by Lemma~\ref{lem:for-rev-flow} we have $f'(x, y) = f(x, y) = 0$. However this contradicts the existence
of the edge $(y, x)$ in $\res$ which must be present because of non-zero flow $f$ on the edge $(x, y)$.
By exchanging $f$ and $f'$ we have: $x \in S_{f'} \implies x \in S_{f}$.

\item The proof of $x \in T_f \iff x \in T_{f'}$ is analogous, except that we need to perform a BFS of $H(f)$ from $t$ by traversing each
edge in the reverse direction.

\end{itemize}

The above two cases immediately imply that $x \in U_f \iff x \in U_{f'}$.
\end{proof}


The next two lemmas, which are specific to our flow network, are useful in proving the rank-maximality of our algorithm 
in the next section.
Consider a flow network $H$ with a max-flow $f$. 
Consider a node $C_a^i \in T \cup U$ such that the predecessor $C$ of $C^i_a$ is in $S$. Such a $C$ must exist since $s\in S$  
and $s$ is an ancestor of $C_a^i$. 

\begin{lemma}\label{lem:O-U-app}
Consider a node $C_a^i \in T \cup U$ such that either the parent $C^j_a$ of $C^i_a$ in $\mathcal{T}_a$  is in $S$ or $C_a^i = C_a^*$.
Then the following hold:
\begin{itemize}
\item (i) Every leaf $C^p_a$ in the subtree of $C_a^i$ in $\mathcal{T}_a$ belongs to $T \cup U$.
\item (ii) Every max-flow $f$ must saturate the edge $(C,C_a^i)$.
\end{itemize}
Conversely, in the applicant-tree, every leaf node $C_a^p \in T \cup U$ has an ancestor $C_a^i \in T \cup U$ 
such that the predecessor $C$ of $C_a^i$ (possibly $s$) is in $S$
and the edge $(C, C_a^i)$ is saturated in every max-flow.
\end{lemma}
\begin{proof}
We prove $(i)$  by arguing that every node in the subtree rooted at $C_a^i$ in $\mathcal{T}_a$ belongs
to $T \cup U$. This immediately implies that every leaf in the subtree of $C_a^i$ belongs to $T \cup U$.
For the sake of contradiction, assume that there exists a
descendant $C_a^{\ell}$ of $C_a^i$ in $\mathcal{T}_a$ such that $C_a^{\ell} \in S$
and let $C_a^{\ell}$ be one of the nearest such descendants of $C_a^i$. Let $C_a^x$ be the parent of $C_a^{\ell}$ in $\mathcal{T}_a$.
We remark that $C_a^x \in T \cup U$ because of the choice of the nearest descendant.
 Since $C_a^{\ell} \in S$ there exists
a path from $s$ to $C_a^{\ell}$ in the residual network $H(f)$. Let $y$ be the last node on the $s \leadsto C_a^{\ell}$
path in the $H(f)$. We note that $y \neq C_a^x$. However, all the other edges incident on $C_a^{\ell}$ in $H$
are outgoing edges from $C_a^{\ell}$. Thus, $y$ is one of the children of $C_a^{\ell}$ in $\mathcal{T}_a$.
Furthermore since $(y, C_a^{\ell})$ appears in $H(f)$, it implies that flow along the edge $(C_a^{\ell}, y)$ is non-zero.
However, we note that the flow along the unique incoming edge $(C_a^x, C_a^{\ell})$ must be zero. If not, the edge $(C_a^{\ell}, C_a^x)$ belongs to $H(f)$ and contradicts the fact that $C_a^x \in  T \cup U$.
However, if the incoming flow to $C_a^{\ell}$ is zero and the outgoing flow from $C_a^{\ell}$ is non-zero, then it contradicts
flow conservation. Therefore such a node $C_a^{\ell} \in S$ does not exist. This finishes the proof of $(i)$.
To prove $(ii)$, we observe that the edge $(C_a^j, C_a^i)$ (or $(s, C_a^i)$ in case $C_a^i = C_a^*$) is a forward edge of the min-cut $(S, U \cup T)$ and hence
must be saturated by every max-flow.

To show the converse, note that $C_a^p \in T \cup U$. Consider the directed path from $s$ to $C_a^p$ in $H$. If no edge on this path
is saturated, then, in $H(f)$, $C_a^p \in S$, a contradiction. Thus along the path $s, C_a^* \ldots C_a^p$, there must exist
an edge $(C,C_a^i)$ such that $C \in S$ and $C_a^i \in T \cup U$. The edge $(C, C_a^i)$ is a forward edge of the $(S, T \cup U)$ min-cut
and hence is saturated by every max-flow of $H$.
Suppose there does not exist a node $C_a^j$ in the path such that $C_a^j \in S$. Then we call $C_a^*$ as $C_a^i$ and the edge $(s,C_a^i)$ is a forward edge of the min-cut $(S,U \cup T)$ and must be saturated by every max-flow. This show that the converse is true.
\end{proof}

\noindent An analogous claim can be proved for the leaf classes in the post-tree:
\begin{lemma}\label{lem:analog-O-U-app} Consider a node $C_p^i \in S \cup U$ such that the parent $C$ of $C^i_p$ in the post-tree is in $T$.
Then the following hold:
\begin{itemize}
\item (i) Every leaf $C^a_p$ in the subtree of $C_p^i$ belongs to $S \cup U$.
\item (ii) Every max-flow $f$ must saturate the edge $(C_p^i, C)$.
\end{itemize}
Conversely, for a leaf node $C_p^a \in S \cup U$, there exists an ancestor $C_p^i \in S \cup U$ 
such that the parent $C$ of $C_p^i$ (possibly $t$) is in $T$
and the edge $(C_p^i, C)$ is saturated in every max-flow.
\end{lemma}

\section{Algorithm for Laminar \CRMM}
\label{sec:pseudocode}

This section gives the detailed pseudo-code for our iterative algorithm for computing a laminar \CRMM\ (see Algorithm~\ref{algo:main}).
At a high level, in each iteration our algorithm operates as follows: it computes a max-flow $f_k$ in a flow network $H_k$ (Step~\ref{step:compute}) and  computes the partition of the vertices $S_k, T_k, U_k$ w.r.t $f_k$ (Step~\ref{step:label}).  
The algorithm then deletes forward and reverse edges of min-cut $(S_k,  T_k \cup U_k)$ (Step~\ref{step:EE-EU}). This step is crucial to ensure that the signature of the matching corresponding to the flow in the subsequent iterations does not degrade.
Finally, the algorithm deletes certain edges of  rank higher than $k$ from the given bipartite  graph (Step~\ref{step:del-higher-rank}) --  we prove that these edges cannot belong to any \CRMM\ and hence can be removed.

We begin by constructing the flow network $H_0$ as described in Section~\ref{sec:flownw}. The max-flow $f_0=0$ in $H_0$ since
there is no $s$-$t$ path in $H_0$.
We partition the edges of $G$ into sets $E_k$, $1\leq k\leq r$ where $r$ is the maximum rank on any edge of $G$ and $E_k$ contains the edges of 
rank $k$ from $G$. Start with $G'_0=G_0=(A\cup P,\emptyset)$. Our algorithm repeatedly constructs the network $H_k$ and maintains the
reduced bipartite graph $G_k'$. Finally the output of our algorithm is the $R$-$L$ edges of the flow network $H_r'$ constructed in the final iteration.

We illustrate these steps on the example in Figure~\ref{fig:example1}. 
Add to $H_0$ (shown in Figure~\ref{new1})
edges of the form $(C_a^p, C_p^a)$ for every rank-1 edge in $G$ to obtain the flow network $H_1$. Let $f_1$ be a max-flow in $H_1$ corresponding to the matching 
$M_1 =  \{(a_1, p_1), ({a_3}, {p_2}), (a_4, {p_5})\}$. That is, for an edge $(a, p) \in M_1$ the unique $s-t$ path containing the edge
$(C_a^p, C_p^a)$ in $H_1$ carries unit flow. 
Figure~\ref{new2} (thick and gray edges) shows the residual network $H_1(f_1)$
along with the partition of the vertices as $S_1, T_1, U_1$. The edge $(C_{p_1}^{a_3}, C_{p_1}^1)$ in $H_1(f_1)$  is an edge of 
the form $(U_1, S_1)$ and hence is deleted as a reverse edge of the min-s-t cut. The edge $(C_{p_1}^*, C_{p_1}^1)$ in $H_1(f_1)$ 
is of the form
$(T_1, S_1)$, however, note that the edge {\em was} a forward edge in $H_1$. Thus we say that $(C_{p_1}^*, C_{p_1}^1)$ is deleted as a forward edge of the min-s-t cut. 
Algorithmically, both these edges
are deleted in Step~\ref{step:EE-EU} of Algorithm~\ref{algo:main}. We denote the flow network obtained after deleting gray edges in Figure~\ref{new2} as $H_1'$. 
Finally, we observe that the edge $(a_2, p_5)$ is a higher rank edge such that $C_{p_5}^{a_2} \in U_1$. Hence this edge is deleted in Step~\ref{step:del-higher-rank} of Algorithm~\ref{algo:main}.
Thus $H_2$ is obtained by adding to $H_1'$ the edges $ (C_{a_1}^{p_4}, C_{p_4}^{a_1}), (C_{a_4}^{p_1}, C_{p_1}^{a_4}), (C_{a_5}^{p_2}, C_{p_2}^{a_5})$.

We remark that if the $(C_{p_1}^*, C_{p_1}^1)$ were not deleted, an augmenting path in $H_2$ of the form 
$\rho_1 = \langle s,C_{a_5}^*,\ldots,C_{p_5}^{*},\ldots, C_{a_4}^{*},C_{a_4}^{p_1},C_{p_1}^{a_4},C_{p_1}^{2},C_{p_1}^{*},C_{p_1}^{1}, \ldots, C_{a_1}^{*}, \ldots ,C_{p_4}^{*},t \rangle$ can be used
to degrade the signature on rank-1 edges. We prove in the subsequent sections that our deletions ensure that the signature is never
degraded.

\begin{algorithm}[ht]
\begin{algorithmic}[1]
\STATE Construct the flow network $H_0 = (V,F_0)$ as described in Section~\ref{sec:flownw}.
\STATE Let $F_0' = F_0$ and for each $i$ set $E'_i = E_i$. 
\FOR{$k=1$ to $ r$}
\STATE $H_{k}=(V,F_k)$ where $F_k=F'_{k-1}\cup \{(C^p_a,C^a_p)\mid (a,p)\in E'_k\}$. \label{step:construct} 
\STATE Let $f_k$ be a max-flow in $H_k$. \label{step:compute} 
Compute the residual graph $\Res$ w.r.t. flow $f_k$.
\STATE Compute the sets $S_k$, $T_k$ and $U_k$. \label{step:label}
\STATE Delete all edges of the form $(T_k\cup U_k,S_k)$ in $H_k(f_k)$. \label{step:EE-EU}
\STATE Delete an edge $(a,p) \in E_j'$ where $j>k$ if   $C_a^p \in T_k \cup U_k$ or $C_p^a \in S_k \cup U_k$. \label{step:del-higher-rank} 
\STATE Let $H'_k = (V, F'_k)$ be the modified $H_k(f_k)$ and let  $G'_k= (A \cup P, \bigcup_{i=1}^k E'_i)$.
 \label{step:red}
\STATE Let $M_k = \{ (a,p) | (C_p^a,C_a^p) \in H'_k \}$.
\ENDFOR
\STATE Return $M_r$.
\end{algorithmic}
\caption{Laminar CRMM}
\label{algo:main}
\end{algorithm}

\begin{lemma}\label{lem:no-LR-del}
Any edge between $C_a^p$ and $C_p^a$ in $H_k(f_k)$ is of the form $S_kS_k$, $T_kT_k$ or $U_kU_k$, irrespective of
its direction in $H_k(f_k)$.
Hence an edge between $L$ and $R$ is never deleted during the course of the algorithm.
\end{lemma}
\begin{proof}
Let $e=(C_a^p, C_p^a)$ be an edge in $H_k$. 
Recall that this is the only outgoing edge for $C_a^p$ and only incoming edge for $C_p^a$ in $H_k$.
Also, $C_a^p$ has an incoming edge of capacity $1$ from its parent and $C_p^a$ has an outgoing edge with capacity
$1$ to its parent.

{\em Case $1$:} Edge $e$ does not carry a flow in $f_k$. Then $C_a^p$ and $C_p^a$ do not receive any flow. In $H_k(f_k)$, $e$ retains
its direction. Thus if $C_a^p$ is in $S_k$, so is $C_p^a$. Conversely, if $C_p^a$ is in $S_k$, then $C_a^p$ has to be in $S_k$, since $C_p^a$
has no other incoming edge, and hence the path from $s$ to $C_p^a$ must use the edge $e$. Similarly, $C_a^p$ is in $T_k$ if and only
if $C_p^a$ is in $T_k$. If $C_a^p$ is in $U_k$, then by the same argument as above, $C_p^a$ can not be in $S_k$ or $T_k$ and hence must be
in $U_k$. 

{\em Case $2$:} Edge $e$ carries a flow of $1$ unit in $f_k$. Then the direction of $e$ is reversed in $H_k(f_k)$, thus $(C_p^a,C_a^p)$ is in
$H_k(f_k)$. Similarly, the direction of the edge to $C_a^p$ from its parent and of the edge from $C_p^a$ to its parent is also reversed.
Thus, both $C_a^p$ and $C_p^a$ still have only one incoming and one outgoing edge in $H_k(f_k)$. Now, if $C_a^p$ is in $S_k$, the only path
possible from $s$ to $C_a^p$ has to be through $C_p^a$ and hence $C_p^a$ must be in $S_k$. Conversely, if $C_p^a$ is in $S_k$, so is
$C_a^p$ since $(C_p^a,C_a^p)\in H_k(f_k)$. An analogous argument holds for containment in $T_k$, and hence in $U_k$ as well.
\end{proof}

\begin{corollary}\label{cor:del}
For every edge $(C_a^p,C_p^a)$ in $H_k$ that carries flow unit flow in $f_k$, either one edge on the path from $s$ to $C_a^p$ in $H_k$
or an edge on the path from $C_p^a$ to $t$ in $H_k$, but not both, is deleted in the $k$-th iteration of Algorithm~\ref{algo:main}.
\end{corollary}
\begin{proof}
By Lemma \ref{lem:no-LR-del}, each edge $(C_a^p,C_p^a)$ has both its end-point in the same set i.e. $S$, $U$, or $T$. If both the end-points
are in $S$, by Lemma~\ref{lem:analog-O-U-app}, an edge on the path from $C_p^a$ to $t$ is deleted in Step $7$ of the algorithm.
We argue that
no edge on the path from $s$ to $C_a^p$ gets deleted.
Let $\rho_A$ be the path from $s$ to $C_a^p$ that carried flow in $H_k$. Then every edge on the path $\rho_A$ is reversed in $H_k(f_k)$
and because $C_a^p \in S$, every vertex on $\rho_{A}$ also belongs to $S$. This implies that no edge on the path $\rho_A$ gets
deleted.

If both the end-points
are in $U$ or $T$, by Lemma \ref{lem:O-U-app}, an edge on the path from $s$ to $C_a^p$ is saturated and hence deleted in Step $7$ of the algorithm. An argument similar to above shows that no edge on the path from $C_a^p$ to $t$ gets deleted in this case.
\end{proof}

\subsection{Rank-maximality of the output}
To prove correctness, we consider flow networks $\tH_i= (V,F_0 \cup \{ (C_a^p,C_p^a) \mid (a,p) \in \bigcup_{j \leq i} E_j \})$
and first establish a one-to-one correspondence between matchings in $G_i$ and
flows in $\tH_i$. With an abuse of notation, we call an edge $(C^p_a,C^a_p)$ in any flow network $H$ a rank $k$ edge
if the corresponding edge $(a,p)$ in $G$ has rank $k$. Also, we refer to directed edges from leaves in the applicant-tree to leaves in the post-tree 
as {\em $L$-$R$ edges} and directed edges from leaves in the post-tree to leaves in the applicant-tree as {\em $R$-$L$ edges}.
In the following lemma, we establish a correspondence between matchings in $G_i$ and flows in $\tH_i$. 
\begin{lemma}\label{lem:flow-mat-corr}
For every feasible matching $M_i$ in $G_i$, there is a corresponding feasible flow $\tf_i$ in $\tH_i$ and vice versa. Moreover, the edges present in $M_i$
are precisely the $L$-$R$ edges in $\tH_i$ that carry one unit flow in $\tf_i$ and hence appear as $R$-$L$ edges in the residual network
$\tH_i(\tf_i)$.
\end{lemma}
\begin{proof}
Let $g_i$ denote a flow in the network $X_i$. Let  $M_i = \{ (a,p) \mid g_i(C_a^p,C_p^a) = 1\}$ be the corresponding matching
constructed using $\tf_i$. It is straightforward to verify that the matching $M_i$ respects the vertex and the class
capacities due to the construction of our flow network.

To prove the other direction let $M_i$ be any feasible matching in $G_i$.  Construct $\tf_i$ as follows:
Start with a flow function $\tf_i$ which assigns
every edge in $\tH_i$ a zero flow. For every edge $(a, p)$ in $M_i$, consider the 
unique path $\rho = \langle s,C_a^*,\ldots,C_a^p,C_p^a,\ldots,C_p^*,t\rangle$ in $\tH_i$. For every edge $e \in \rho$,
increment the flow $\tf_i(e)$ by one. We argue that $\tf_i$ is feasible in $\tH_i$. 
For any class node $C_p^u$, the matching $M$ assigns  $|M(C_p^u)|$ applicants to the class.
Thus the edge $(C_p^u, C_p^v)$ belongs to exactly $|M(C_p^u)|$ such paths. Here $C_p^v$ is the parent of $C_p^u$ in $\mathcal{T}_p$.
Therefore, $\tf_i(C_p^u, C_p^v) = |M(C_p^u)| \le q(C_p^u)$.
Since this holds for class vertex, we conclude that $\tf_i$ is a feasible flow in $\tH_i$.
\end{proof}

We define signature of a flow to be the signature of the corresponding matching in $G$.
\begin{definition}[Rank-maximal flow]
We call a flow $\tf_i$ in a network $\tH_i$ to be {\em rank-maximal} if the corresponding matching $M_i$ is rank-maximal in $G_i$.
\end{definition}
Thus $\tf_i$ is a rank-maximal flow in $\tH_i$ if it uses the maximum number of rank $1$ edges, subject to that, maximum number of rank $2$
edges and so on. 
By flow-decomposition theorem (see e.g. \cite{AMO93}), a flow $\tf_i$ in $\tH_i$ can be decomposed into flow on $s-t$ paths, such that each path uses
exactly one $L$-$R$ edge. Thus, based on the ranks of the $L$-$R$ edges used, $\tf_i$ can be decomposed into flows $\tf_i^1,\ldots,\tf_i^i$ such that,
for each $j$: $1 \le j \le i$, $\tf_i^j$
uses paths only through $L$-$R$ edges of rank $j$. Thus $\tf_i=\tf_i^1+\ldots +\tf_i^i$. We call $\tf_i^j$ to be the $j$th component of $\tf_i$.


%

\begin{lemma}\label{lem:induction}
Suppose, for each $j\leq i$, the $j$th component $\tf_i^j$ of every rank-maximal flow $\tf_i$ in $\tH_i$ is a max-flow in $H_j$. Then the $(i+1)st$ component
$\tf_{i+1}^{i+1}$ of any rank-maximal flow $\tf_{i+1}$ in $\tH_{i+1}$ is a max-flow in $H_{i+1}$.
\end{lemma}
\begin{proof}
The statement clearly holds for $i=1$, since $H_1$ is same as $\tH_1$. Now assume the statement for all $j\leq i<r$. We will prove it for $i+1$.
Moreover, by the definition of rank-maximal flow, $\tf_{i+1}^1+\ldots+\tf_{i+1}^i$ is a rank-maximal flow in $\tH_i$, call it $\tf_i$.

Let $e$ be an edge with residual capacity $c>0$ in $\tH_i$ when the flow $\tf_i$ is set up in $\tH_i$. We show that $e$ has the same
residual capacity in $H_i(\tf_{i+1}^i)$, and hence in $H_{i+1}$. This clearly holds in $H_1(\tf_{i+1}^1)$ since $H_1$ and $\tH_1$ are the
same networks. Inductively, each $\tf_{i+1}^j$ is a flow in $H_j$ for $1\leq j<i$ and hence the same amount of flow is sent through $e$ in
$\tH_j$ as the total flow sent in $H_1,\ldots,H_j$. Hence the residual capacity of $e$ is the same in $\tH_i(\tf_{i+1}^i)$ as in $H_i(\tf_{i+1}^i)$.

Consider a path $\rho$ in $\tH_{i+1}$ that carries a flow of one unit from $\tf_{i+1}^{i+1}$. Let $e_{\rho}$ be
the rank $i+1$ $L$-$R$ edge on $\rho$. Moreover $\rho_A$ and $\rho_P$ be the subpaths of $\rho$ from $s$ to the leaf node in applicant-tree
and from the leaf node to $t$ in the post-tree.

Every edge $e$ on $\rho$ must be unsaturated by $\tf_{i+1}^1+\ldots+\tf_{i+1}^i$. If this is not the case, then $\tf_{i+1}^{i+1}$
can not be routed through $e$ without reducing some flow from $\tf_{i+1}^1+\ldots+\tf_{i+1}^i$ and the resulting flow will not be rank-maximal.
Since each $\tf_{i+1}^j$ for $1\leq j\leq i$ is a max-flow in $H_j$, and all the edges on $\rho_A$ and $\rho_P$ are unsaturated in
each of the flows, every node on $\rho_A$ is in $S$ and each node on $\rho_P$ is in $T$ in each of the first $i$ iterations of the algorithm. Thus no edge of $\rho_A$ or $\rho_P$ is deleted from $H_j$ in the $j$th iteration of the 
algorithm for any $1\leq j\leq i$, and also, $e_\rho$ is not deleted in Step $7$ in any iteration.

Thus, in the flow-decomposition of $\tf_{i+1}$, every path that carries some flow along a rank $i+1$ edge, is also present in $H_{i+1}$. Moreover,
if $c$ such paths pass through an edge $e$, then as proved above, $e$ has a capacity $c$ in $H_{i+1}$.
Hence $\tf_{i+1}^{i+1}$ is a valid flow in $H_{i+1}$. It has to be a max-flow in $H_{i+1}$, otherwise $\tf_{i+1}$ will not be a rank-maximal flow in
$\tH_{i+1}$.
\end{proof}

\begin{lemma}\label{lem:signature}
Define $Y_i$ as the set of $R$-$L$ edges in $H'_i$. For every $i,j, j>i$, the number of edges of rank at most $i$ is the same in
$Y_i$ and $Y_j$. 
\end{lemma}
\begin{proof}
By Corollary \ref{cor:del}, for each rank $i$ $L$-$R$ edge $(C_a^p,C_p^a)$ that carries a flow and hence becomes an $R$-$L$ edge in $H'_i$, 
either an edge in the path from $s$ to $C_a^p$ or an edge on the path from $C_p^a$ to $t$ is deleted. Moreover, a node that loses the edge to or
from its parent in iteration $i$ never gets edges of rank more than $i$ on any leaf node in its subtree. Without loss of generality, let $C_a^\beta$ be such a node where $a$ is an applicant and $\beta$ is one of the classes of $a$'s classification. Then every augmenting path $\rho$ in the subsequent 
iterations that involve $C_a^\beta$ is of the form $\langle s,\ldots, C_{p'}^a, C_a^{p'} \ldots, C_a^\beta, C_a^{p''}, C_{p''}^a, \ldots, t \rangle$. That is, every augmenting path involving $C_a^\beta$ goes from $s$ to a leaf
in the subtree of $C_a^\beta$ through an $R$-$L$ edge, then it goes to $C_a^\beta$, then to another leaf in its subtree and finally to $t$ through an $L$-$R$ edge incident on that leaf. Thus, augmentation along this path changes the $L$-$R$ edge to $R$-$L$ edge and vice versa, thereby maintaining
the number of $R$-$L$ edges in the subtree of $C_a^\beta$. Since no leaf in the subtree of $C_a^\beta$ has an edge of rank more than $i$ incident
on it, the number of $R$-$L$ edges of rank at most $i$ in the subtree of $C_a^\beta$ is also preserved.

Now it remains to prove that no $R$-$L$ edge of rank at most $i$ is counted twice in the above counting, once from the trees of each of its end-points. 
For this, we show that, if a node $C_a^\beta$ in the applicant-tree and a node $C_p^\alpha$ in the post-tree get the edge to their respective parent
deleted in the $i$th iteration, then there is no directed path between them that uses an edge between the leaves in their respective subtrees.
Thus, if there is an edge between leaf classes $C_a^p$ and $C_p^a$ respectively in the subtrees of $C_a^\beta$ and $C_p^\alpha$, it can not be
used by an augmenting path $\rho$ described above. This is because of the following:

The node $C_a^\beta$ must be in $T_i\cup U_i$ and $C_p^\alpha$ must be in $S_i$
since the edge between them and their respective parent was deleted in iteration $i$. Hence at the end of iteration $i$, there is no directed path
from $C_p^\alpha$ to $C_a^\beta$, otherwise $C_a^\beta$ would be in $S_i$. If there is a directed path from $C_a^\beta$ to $C_p^\alpha$ in $H_i(f_i)$, one of
the edges on that path must have been deleted, since the path is from a node in $T_i\cup U_i$ to a node in $S_i$, and hence an edge on
the path must have one end-point in $T_i\cup U_i$ and another end-point in $S_i$. Hence an augmenting path $\rho$ as described above can not
go directly from $C_a^\alpha$ to $C_p^\beta$ or the other way, without going through other applicant or post trees. Hence $\rho$ can not use
an $R$-$L$ or $L$-$R$ edge between the leaves in the subtrees of $C_a^\beta$ and $C_p^\alpha$.
This shows that the number of $R$-$L$ edges in $Y_i$ does not change in any subsequent iteration.
\end{proof}

Let $f_i$ be a max-flow in $H_i$ and $H_i(f_i)$ denote the corresponding residual network. Let $Y$ denote the set of $R$-$L$ edges
in $H_i(f_i)$. Corresponding to the $R$-$L$ edges in $Y$, we can set up a flow $\tf_i$ which is a feasible flow in $X_i$. To obtain such a flow, we start with every edge having $\tf_i(e) = 0$. Repeatedly select an unselected edge $e$ from $Y$. Let $\rho_e$
denote the 
unique $s-t$ path  in $\tH_i$ containing $e$. We increase the flow along every edge in $\rho_e$ by one unit.
Using arguments similar to Lemma~\ref{lem:flow-mat-corr} we conclude that $\tf_i$ is a feasible flow in $\tH_i$.
\begin{lemma}\label{lem:crmm}
For every $1\leq k\leq r$, the following hold:
\begin{enumerate}
\item For every rank-maximal flow $\tf_k=\tf_k^1+\ldots+\tf_k^k$ of $\tH_k$, $\tf_i$ is a max-flow in $H_i$ for $1\leq i\leq k$.
\item Conversely, the flow $\tf_k$ (constructed as above) corresponding to the $R$-$L$ edges of $H_k(f_k$)  is a rank-maximal flow in $\tH_k$. 
\end{enumerate}
\end{lemma}
\begin{proof}
We prove this by induction on $k$. When $k=1$, $\tH_1$ and $H_1$ are the same networks. A rank-maximal flow $\tf_1$ in $\tH_1$ is 
just a max-flow in $\tH_1$ and hence in $H_1$. Algorithm~\ref{algo:main} also computes a max-flow in $H_1$. Hence both the statements
hold for $k=1$.

Assume the statements to be true for each $j\leq i$. We prove them for $i+1$. The first statement follows from Lemma \ref{lem:induction}.
We prove the second statement. By induction hypothesis, $\tf_i$ corresponding to $f_i$ is a rank-maximal flow in $\tH_i$, let its signature be $(\sigma_1,\ldots, \sigma_i)$. Let the signature of a rank-maximal flow in $\tH_{i+1}$ be $(\sigma_1,\ldots, \sigma_{i+1})$. By Lemma \ref{lem:signature}, the number of $R$-$L$ edges of rank $j$
 in $H'_{i+1}$
and hence in $H_{i+1}(f_{i+1})$ is the same as in $H'_i$, for each $j\leq i$. Thus the signature of $\tf_{i+1}$ in $\tH_{i+1}$ corresponding 
to $f_{i+1}$ is $(\sigma_1,\ldots,\sigma_i,\sigma'_{i+1})$ where $\sigma'_{i+1}\leq \sigma_{i+1}$. However, by Lemma \ref{lem:induction},
the $(i+1)$st component of a rank-maximal flow in $\tH_{i+1}$ is a max-flow in $H_{i+1}$. Since $f_{i+1}$ is also a max-flow in $H_{i+1}$
it must be of the same value and hence the corresponding flow $\tf_{i+1}$ of $f_{i+1}$ must have signature $(\sigma_1,\ldots,\sigma_{i+1})$. 
\end{proof}

\noindent {\bf Running time:} The size  of our flow network is determined by the total number of classes. Due to the tree structure of $T_u$, the size of the flow network is equal to the total size of all preference lists which is $O(|E|)$. The  maximum matching size in our instance is upper bounded by $|E|$ and the max-flow in our network is also at most $O(|E|)$. This gives an upper bound of $O(|E|^2)$ on the running time. Thus we establish Theorem~\ref{thm:poly}.

\section{Classified Popular matchings}\label{sec:pop}
In this section, we address the notion of popularity, an alternative notion which has been well-studied
in the context of one-sided preference lists. We consider the problem of computing a popular matching in the many-to-one setting
with laminar classifications, if one exists, referred to as the \LCPM\ problem here onwards. 
The same problem without classifications has been considered by Manlove and Sng~\cite{MS06} 
as the {\em capacitated house allocation problem with ties} (\CHAT). 

Let $G = (A\cup P, E)$, along with quotas and laminar classifications for each post be the
given \LCPM\ instance. Introduce a unique last resort post $\ell_a$ for each $a\in A$ as the last choice of $a$. Call the modified instance $G$. 
A simple modification of our algorithm from Section~\ref{sec:algo} outputs a popular matching in a given \LCPM\ instance (if it exists) in $O(|A|\cdot |E|)$ time. The correctness proof
of the algorithm also gives the characterization of popular matchings in an \LCPM\ instance. The main steps in the algorithm that computes a popular matching amongst feasible matching (if one exists) are as follows:
\begin{algorithm}[ht]
\begin{algorithmic}[1]

\STATE Construct the flow network $H_0 = (V, F_0)$ as described in Section~\ref{sec:flownw}.\\
\STATE Define $f(a)=$ set of rank-$1$ posts of $a$. 
\STATE Let $H_1=(V,F_1)$, where $F_1=F_0\cup\{(C_a^p,C^a_p)\mid p\in f(a)\}$. 
\STATE \label{step:start} Let $f_1$ be a max-flow in  $H_1$ and let $H_1(f_1)$ be the corresponding residual network.
\STATE Define the sets $L$ and $R$ as
\begin{eqnarray*}
L = \{C_a^p \mid a \in A \mbox { and } p \in N(a) \}; & &R  =  \{ C_p^a \mid  p \in P \mbox { and } a \in N(p) \}
\end{eqnarray*}
\vspace{-0.15in}
\STATE Compute the sets $S_1,T_1,U_1$.
\STATE \label{step:delete} Delete edges of the  form $(T_1\cup U_1, S_1)$ in $H_1(f_1)$. Rename the remaining edges as $F_1'$.
\STATE For each $a$ such that $C_a^* \in S_1$, let $s(a)=$ the set of most preferred posts $p$  of $a$ such that $C_p^a \in T_1$.\\
\COMMENT {Note that $s(a)\neq \emptyset$ due to the last resort post $\ell_a$. }
\STATE \label{step:end} Let $H_2=(V,F_2)$ where $F_2=F_1'\cup \{(C_a^p,C_p^a)\mid C_a^*\in S_1, p\in s(a)\}$.
\STATE Let $f_2$ be a max-flow in $H_2$ and let $H_2(f_2)$ be the corresponding residual network.
\STATE Let $M = \{(a, p) \mid (C_p^a, C_a^p) \in H_2(f_2)\}$.
\STATE If $|M| = |A|$, return $M$, else return ``No popular matching".
\end{algorithmic}
\caption{Laminar CPM}
\label{algo:main-pop}
\end{algorithm}

\REM{
We make the following observation about the $R$-$L$ edges in $H_1(f_1)$.

\begin{claim}
\label{claim1}
Let $(C_p^a, a)$ be an $R$-$L$ edge in $H_1(f_1)$. Then exactly one of the following hold:
\begin{enumerate}
\item $a \in U_1 \cup T_1$. In this case we say that the edge $(C_p^a, a)$ is associated with $a$. \\
OR
\item There exists a unique ancestor $C_p^u$ of $C_p^a$ such that $C_p^u \in S_1$ and the parent $C_p^v$ of $C_p^u$ belongs to $U_1 \cup T_1$. In this
case we say that the edge $(C_p^a, a)$ is associated with $C_p^u$.
\end{enumerate}
This implies that for every $R$-$L$ edge in $H_1(f_1)$ exactly one of the two edges $(a, s)$ or $(C_p^v, C_p^u)$  gets deleted in $H_1(f_1)$.
Conversely if an edge $(a, s)$ gets deleted then there is a exactly one edge $(C_p^a, a)$ associated with it. Similarly, if
an edge $(C_p^v, C_p^u)$ gets deleted as a forward edge in $H_1(f_1$), then there are exactly $q(C_p^u)$ $R$-$L$ edges associated with $C_p^u$.
\end{claim}
}

\subsection{Correctness and characterization of classified popular matchings}
We show that the algorithm described above outputs a popular matching, and thereby, give a characterization of popular matchings similar to that of 
Abraham~et~al.~\cite{AIKM07} and  \cite{MS06}.

\begin{lemma}
\label{lem:max-rank-1}
Let $M$ be a popular matching amongst all the feasible matchings in a given \LCPM\ instance $G$. 
Then the max-flow $f_1$ in $H_1$ has value $|M \cap E_1|$.
\end{lemma}
\begin{proof}
Let $M_1 = M \cap E_1$. Note that $M_1$ is feasible in $G$ since $M$ is feasible in $G$. Therefore,
$M_1$ has a corresponding flow $f'_1$ in $H_1$. Hence the max-flow $f_1$ in $H_1$ has value at least $|M_1|$.
For contradiction, assume that $f_1$ has value strictly larger than $|M_1|$. We show how to obtain a feasible matching that is more popular
than $M$, contradicting the popularity of $M$. 

Since $f'_1$ is not a max-flow in $H_1$, there exists an augmenting path w.r.t. $f'_1$ in $H_1$.
Let $\rho = \langle s, C_{a_1}^*, C_{a_1}^{p_1},  C_{p_1}^{a_1}, \ldots,$ $ C_{a_j}^*, C_{a_j}^p, C_p^{a_j}, C_p^{1}, C_p^{2}, \ldots , C_p^*, t \rangle$ be the augmenting path. 
Let the last node from $R$ present on $\rho$ be $C_p^{a_j}$.
The subpath of $\rho$, denoted as $tail (\rho)$, is the subpath from $C_p^{a_j}$ to its ancestor $C_p^*$. 
Here $(a_j, p) \in E$. Clearly, every node $C_p^u  \in tail(\rho)$ is such that $|M_1(C_p^u)| < q(C_p^u)$, that $C_p^u$ is  under-subscribed in $M_1$. We consider two cases:
\begin{itemize}
\item {\em Every node  $C_p^u \in tail(\rho)$ is under-subscribed in $M$: }
In this case, we can augment the flow $f_1'$, 
and hence modify the matching $M_1$ and consequently $M$, 
to match applicant $a_1$ to its rank-$1$ post. Note that the rest of the applicants on $\rho$
continue to be matched to their rank-$1$ post since the augmentation is done using only rank-$1$ edges.
Thus we obtain a matching $M'$ that is more popular than $M$, a contradiction.
\item There exists some node $C_p^u \in tail (\rho)$ such that $|M( C_p^u)| = q(C_p^u)$. Consider such a class node $C_p^u \in tail(\rho)$ that is
nearest to $C_p^{a_j}$. Let $a_k \in M(C_p^u)$ be such that $a_k$ treats $p$ as a non-rank-$1$ post. Such an applicant $a_k$ must exist 
because $C_p^u$ is not saturated w.r.t. $f_1'$ (since the augmenting path exists in $H_1$) but $C_p^u$ is saturated in $M$.   
Recall $M_1 = M \cap E_1$ and let $M_2  = M \setminus M_1$.
Construct the matching $\hat{M} = M_1 \cup (M_2 \setminus \{(a_k, p)\})$. With respect to $\hat{M}$, every node on $tail(\rho)$ is under-subscribed. 
Now we are in the similar case as above and we  can augment $f_1'$ along $\rho$ to get $M_1'$. In $M_1'$, apart from $a_1$ which gets matched to its rank-1 post $p_1$,
every other applicant on $\rho$ continues to be matched to one of its rank-1 posts.  
Now, $M' = M_1' \cup M_2 \setminus \{(a_k, p)\}$
and  $M'(a_1) = p_1$. Note that for any post $p' \neq p$, for any class node $C_{p'}^u$, we have $|M(C_{p'}^u)| = |M'(C_{p'}^u)|$ 
and hence 
$M'$ is a feasible matching in $G$.

Finally
consider any $p' \in f(a_k)$ and let $Y = \langle C_{p'}^{a_k}, C_{p'}^1, \ldots, C_{p'}^* \rangle$ denote the unique path from
$C_{p'}^{a_k}$ to $C_{p'}^*$ in $\mathcal{T}_{p'}$. If every class $C_{p'}^j \in Y$ is such that $|M(C_{p'}^j)| < q(C_{p'}^j)$ then
we can construct $N = M' \cup \{(a_k, p')\}$. Here, both $a_1$ and $a_k$ prefer $N$ over $M$ a contradiction to the popularity of
$M$.  Thus, in this case we are done with the proof.
Assuming we do not fall in the above case, there must exist a class node $C_{p'}^u \in Y$ such that 
$|M(C_{p'}^u)| = q(C_{p'}^u)$
and let $C_{p'}^u$ denote the nearest  such class from $C_{p'}^{a_k}$.  Let $a_t \in M(C_{p'}^u)$. Construct the matching $N = M' \setminus \{(a_t, p')\} \cup \{(a_k, p')\}$.
The matching $N$ is feasible in $G$ and both $a_1$  and $a_k$ prefer $N$ to $M$ whereas the applicant $a_t$ prefers $M$ to $N$. Thus
we have obtained a feasible matching that is more popular than $M$, a contradiction.
\end{itemize}
This completes the proof of the lemma.
\end{proof}

We now show that, in a popular matching, every applicant $a$ has to be matched to a post belonging to  $f(a) \cup s(a)$. For the sake of brevity,
we refer to a post $p$ {\em an $f$-post} (respectively {\em an $s$-post}) if there is an applicant $a$ such that $p\in f(a)$ (respectively, $p\in s(a)$).

\begin{lemma}
\label{lem:bet-fa-sa}
Let $M$ be a popular matching amongst all feasible matchings in an \LCPM\ instance $G$, then for any $a \in A$, $M(a)$ is never
strictly between $f(a)$ and $s(a)$. 
\end{lemma}
\begin{proof}
 For contradiction, assume that $M(a) = p$ and $p$ is strictly between $f(a)$ and $s(a)$.  
Since $p \notin s(a)$,
it implies that $C_p^a \in S_1 \cup U_1$ with respect to the max-flow $f_1$ in $H_1$. By converse of Lemma~\ref{lem:analog-O-U-app} for posts,
we claim that there must exist an  ancestor $C_p^u$ of $C_p^a \in \mathcal{T}_p$ such that $C_p^u \in S_1 \cup U_1$ and its  parent $C_p^v \in T_1$.
Thus by Lemma~\ref{lem:for-rev-flow} (a), the edge $(C_p^u, C_p^v)$ must  be
saturated w.r.t. every max-flow of $H_1$. This implies that in the matching $N$ corresponding to any max-flow in $H_1$, we have $|N(C_p^u)| = q(C_p^u)$.
Consider the flow $f_1'$ corresponding to $M_1 = M \cap E_1$ in $H_1$. By Lemma~\ref{lem:max-rank-1}
$f_1'$ must be a max-flow in $H_1$. Thus $|M_1 (C_p^u)| = q(C_p^u)$. Note that $M(a) = p$ and $a$ does not treat $p$ as its rank-1 post.
Thus for $M$ to be feasible, it must be the case that  $|M_1(C_p^u)| < q(C_p^u)$, a contradiction.
This completes the proof that $M(a)$ cannot be strictly between $f(a)$ and $s(a)$.
\end{proof}
\begin{lemma}
\label{lem:worse-than-sa}
Let $M$ be a popular matching amongst all feasible matchings in an \LCPM\ instance $G$, then for any $a \in A$, $M(a)$ is never
strictly worse than $s(a)$.
\end{lemma}
\begin{proof}
 Assume that $M(a) = p$ where $p$ is strictly worse than $s(a)$ on the preference list of $a$. If there exists a post $p' \in s(a)$ such
that every node on the path from $C_{p'}^a$ to $C_{p'}^*$ in $\mathcal{T}_{p'}$ is under-subscribed in $M$, then we are done. This is because
we can construct a feasible matching $M' = M \setminus \{(a, M(a)\} \cup \{(a, p')\}$ which is more popular than $M$, completing the proof.

Thus it must be the case that, for every $p' \in s(a)$, some node $C_{p'}^u$ in the path mentioned above is saturated in $M$. Moreover, let $C_{p'}^u$
be the class closest to $C_{p'}^a$ in $\mathcal{T}_{p'}$ that is saturated in $M$.
Let $a' \in M(p')$ such that both and $a$ and $a'$ belong to $C_{p'}^u$. 
We break the proof into two parts based on whether $a'$ treats $p'$ as a rank-$1$ post or as a non-rank-$1$ post.
\begin{itemize}
\item  {\em Applicant $a'$ treats $p'$ as a non-rank-1 post:} 
In this case, we can construct another matching $M'=(M\setminus \{(a,p),(a',p')\})\cup\{(a,p'),(a',p'')\}$ where $p''\in f(a')$. 
If $M'$ does not exceed the quota of any class of $p''$, we are done, since both $a$ and $a'$ prefer $M'$ over $M$.

In case $M'$ exceeds quota of some class of $p''$ 
containing $a'$, we pick an arbitrary applicant $b\neq a'$ from $M(p'')$ such that $b$ belongs to the class closest to $C^{a'}_{p''}$ in $\mathcal{T}_{p''}$ whose quota is exceeded in $M'$ and reconstruct $M'$
as $M'=(M\setminus \{(a,p),(a',p'),(b,p'')\})\cup\{(a,p'),(a',p'')\}$. Clearly, $M'$ is feasible in $G$. Also, $a,a'$ prefer $M'$ over $M$ whereas only $b$ prefers $M$ over $M'$.
Therefore $M'$ is more popular than $M$, contradicting the assumption about the popularity of $M$.
\item {\em Applicant $a'$ treats $p'$ as a rank-$1$ post: }Since $p' \in s(a)$, it implies that $C_{p'}^a \in T_1$ in $H_1$. That is,
there is a path $\rho$ from $C_{p'}^a$ to $t$ in the residual network $H_1(f_1)$. 
In this case, we use arguments similar to Lemma~\ref{lem:max-rank-1} to come up with a matching more popular than $M$.
\end{itemize}
This completes the proof of the lemma.
\end{proof}

\begin{lemma} Let $M$ be a feasible matching in an \LCPM\ instance $G$. The matching $M$ is popular amongst feasible matchings in $G$ if and only if $M$ satisfies the following two properties:
\begin{itemize}
\item $M \cap E_1$ has a max-flow corresponding to it in $H_1$, and 
\item For every $a \in A$, $M(a) \in f(a) \cup s(a)$.
\end{itemize}
\label{lem:char}
\end{lemma}
\begin{proof}
The necessity of the above properties has already been shown. We now show that they are sufficient.
Let $M$ be a feasible matching that satisfies both the conditions of the lemma and for contradiction assume that $M$ is not popular amongst feasible matchings in $G$.
Let $M'$ be a feasible matching more popular than $M$ and let $a$ be an applicant that prefers $M'$ over $M$. 
Our goal is to show that for each $a$ there exists a unique applicant $b$  that prefers
$M$ over $M'$. 

Since $a$ prefers $M'$ over $M$, it implies that $M(a)=p$ is not a  rank-$1$ post for $a$. 
Furthermore since $M(a) \in s(a)$ (as $M$ satisfies the conditions of the lemma) and $M'(a) = p'$ it implies that $C_{p'}^{a} \in S_1 \cup U_1$ in $H_1$. 

Consider the node $C_{p'}^{a}$. Observe that $a \in M'(p') \setminus M(p')$ by choice of $a$. 
We claim that there exists some applicant $ a_1 \in  M(p') \setminus M'(p')$ such that $p' \in f(a_1)$.  
Since $C_{p'}^{a} \in S_1 \cup U_1$ and by converse of Lemma~\ref{lem:analog-O-U-app}
there exists an ancestor  $C_{p'}^u$ of $C_{p'}^{a}$ which is saturated w.r.t. $f_1$. 
If $a \in C_{p'}^u$, then  since $C_{p'}^u$ is saturated w.r.t. the flow $f_1$ there is an applicant $a_1 \in C_{p'}^u$ such that $M(a_1) \in f(a_1)$ and 
$M'(a_1) \neq M(a_1)$. 
Otherwise $a \notin C_{p'}^u$. Again if $M(C_{p'}^u) \neq M'(C_{p'}^u)$ we find the desired applicant $a_1 \in M(C_{p'}^u) \setminus M'(C_{p'}^u)$. Therefore assume  that  $M(C_{p'}^u)  = M'(C_{p'}^u)$. However, note that $M$ restricted to rank-1 edges
is a max-flow in $H_1$. Since $M'(p')$ has at least one more applicant matched along rank-1 edges (that is the applicant $a$), 
it implies that  
there is some applicant $a_1$ such that $M(a_1) \neq M'(a_1)$ and $M(a_1) \in f(a_1)$.
If $a_1$ is not matched
to a rank-1 post in $M'$ we are done, since $a_1$ is our desired applicant $b$. 

Else we consider $p_1 = M'(a_1)$. We claim that the node $C_{p_1}^{a_1} \notin T_1$. Otherwise the path $\langle C_{p'}^{a}$\\ $\ldots C_{p'}^{u} \ldots  C_{p'}^{a_1} \ldots  C_{p_1}^{a_1} \ldots t \rangle$ shows that $C_{p'}^a \in T_1$
a contradiction to the fact that $C_{p'}^{a} \in S_1 \cup U_1$. Thus $C_{p_1}^{a_1} \in S_1 \cup U_1$. We now find an applicant $a_2 \in M(p_1) \setminus M'(p_1)$ such that $p_1 \in f(a_2)$ and $a_2 \neq a_1 \neq a$.
Again if $a_2$ is not matched to a rank-1 post in $M'$ we are done since $a_2$ is the desired applicant $b$.
We note that our exploration which has started at $C_{p'}^{a}$ must find these distinct applicants $a_1, a_2, \ldots, a_k$ since
the corresponding post nodes were in $S_1 \cup T_1$. We also note that the applicant $a$ cannot be one of the $a_i$, $1 \le i \le k$
since $a$ is not matched to a rank-1 post in $M$. 
Thus the exploration terminates at an applicant $a_k$ such that $M(a_k) \in f(a_k)$ and $M'(a_k) \notin f(a_k)$.
The applicant $a_k = b$ is the desired applicant which prefers $M$ over $M'$.

Note that we need to ensure that for every $a$ there is a unique $b$ such that the votes are compensated. Hence for 
another applicant $a'$ which prefers $M'$ over $M$, we use the same arguments as above, except that we do not consider
any applicant that was already used in a prior exploration. We are guaranteed to find such an applicant, since the corresponding post node
is in $S_1 \cup U_1$, implying that some ancestor of the node is saturated w.r.t. the max-flow $f_1$.
This completes the proof.
\end{proof}

\begin{lemma}
Let $M$ be the matching produced by Algorithm~\ref{algo:main-pop}. Then $M$ satisfied both the conditions of Lemma~\ref{lem:char}. \end{lemma}
\begin{proof}
We first prove that
the number of rank-1 edges in $M$ is equal 
to the value of  max-flow in $H_1$.
Let $f_1$ be the max-flow 
$H_1$ and by max-flow min-cut theorem, the value of $f_1$ is equal to the sum of capacities of the forward edges of the min-s-t cut ($S_1, U_1 \cup T_1)$. Thus,
\begin{eqnarray*}
|f_1| = \sum_{(x, y) \in H_1: x \in S_1, y \in U_1 \cup T_1} c(x,y)
\end{eqnarray*}
We observe that such an edge $(x, y)$ appears as a $(y, x)$ edge in the residual network $H_1(f_1)$ and gets deleted during Step~\ref{step:delete} of our algorithm. 
We note that by Lemma~\ref{lem:no-LR-del} no edge between $L$ and $R$ is deleted by our algorithm.
Therefore, an $(x, y)$ edge in $H_1$ whose corresponding $(y, x)$ edge gets deleted in $H_1(f_1)$ has to be either of the two types:
\begin{itemize}
\item $(x, y) = (s, C_a^*)$ for some applicant $a$. In this case $c(x, y) = 1$.
\item $(x, y) = (C_p^u, y)$ for some post $p$. In this case $c(x, y) = q(C_p^u)$.
\end{itemize}
The node $C_a^*$ is saturated in $f_1$ thus there is exactly one $R$-$L$ edge incident on $C_a^*$ in $H_1(f_1)$. 
Similarly, the node $C_p^u$ is saturated in $f_1$ and hence in $H_1(f_1)$ there are exactly $q(C_p^u)$ rank-1 $R$-$L$ edges 
in the subtree of $C_p^u$. By Corollary~\ref{cor:del} an $R$-$L$ edge is counted for either $C_a^*$ or $C_p^u$ but not both.

We show that (i) for an applicant $a$ if the edge $(C_a^*, s)$ got deleted, then  there is one rank-1 $R$-$L$ edge in the subtree of $C_a^*$ in $H_2(f_2)$ and 
(ii) for a node $C_p^u$ if the edge $(y, C_p^u)$ got deleted, then there are $q(C_p^u)$ many rank-1 $R$-$L$ edges in the subtree of $C_p^u$ in $H_2(f_2)$.
\begin{itemize}
\item Consider the node $C_a^*$. Since $C_a^* \in U_1 \cup T_1$, no node in the subtree of $C_a^*$ gets any non rank-1 edges on it during construction of $H_2$.
If $f_2$ does not use the node $C_a^*$, the $R$-$L$ edge in the subtree of $C_a^*$ in $H_1(f_1)$  continues to exist in $H_2(f_2)$ and we are done.
If the flow $f_2$ uses the node $C_a^*$, since the edge $(C_a^*, s)$ is deleted, the flow must be via a path of the form $\langle \ldots C_p^a, C_a^p, C_a^*, C_a^{p'}, C_{p'}^a, \ldots\rangle $.
Note that $p'$ is a rank-1 post of $a$ and hence $(C_{p'}^a, C_a^{p'})$ is the $R$-$L$ edge in the subtree of $C_a^*$ in $H_2(f_2)$.

\item Consider the node $C_p^u$. Let $\mathcal{T}(C_p^u)$ denote the subtree of $\mathcal{T}_p$ rooted at $C_p^u$.
Since $C_p^u \in S_1$, by Lemma~\ref{lem:analog-O-U-app} every leaf in $\mathcal{T}(C_a^v)$ is in $S_1 \cup U_1$. 
Thus, none of the leaf nodes in $\mathcal{T}(C_p^u)$ is $s(a')$ for any applicant $a'$. Thus, none of these nodes get non-rank-1 edges incident on them. 
If $f_2$ does not use $C_p^u$ we are  done, since the $q(C_p^u)$ many rank-1 $R$-$L$ edges in $\mathcal{T}(C_p^u)$ from $H_1(f_1)$ continue to exist in $H_2(f_2$).
If $f_2$ uses $C_p^u$, since the edge $(C_p^v, C_p^u)$ is deleted (by our algorithm), the flow must enter and exit via leaves of the subtree $\mathcal{T}(C_p^u)$. This implies that for
every $R$-$L$ edge via which the flow enters to reach $C_p^u$, there must be a unique $L$-$R$ edge via which the flow leaves $C_p^u$. This ensures that the number of rank-1
$R$-$L$ edges in the subtree of $\mathcal{T}(C_p^u)$ remains invariant between $H(f_1)$ and $H(f_2)$.
\end{itemize}
To complete the proof we argue that a single $R$-$L$ edge in $H_2(f_2)$ does not get counted for an applicant node $C_a^*$ {\it and} for a post node $C_p^u$.
Assume for the sake of contradiction, an $R$-$L$ edge $(C_p^a, C_a^p)$ is counted for both $C_a^*$ and $C_p^u$. This implies that in $f_2$, there is unit
flow along the edge $(C_a^p, C_p^a)$. 
Let the flow via $(C_a^p, C_p^a)$ in $f_2$  be along the path $ \rho =  \langle s, \ldots, C_a^*, C_a^p, C_p^a, \ldots, C_p^u, \ldots, t \rangle$.
Recall that since $(C_a^*, s)$ was deleted, the node $C_a^* \in T_1 \cup U_1$. Lemma~\ref{lem:O-U-app} implies that $C_a^p \in T_1 \cup U_1$.
Similarly, $C_p^u \in S_1$ implies that $C_p^a \in S_1 \cup U_1$. The path $\rho$ must have some edge $(x, y)$ such that $x \in T_1 \cup U_1$ and $y \in S_1$. However, all such edges from $T_1 \cup U_1$ to $S_1$ were deleted by 
our algorithm. Thus the path $\rho$ does not exist in $H_2$ which implies that a single $R$-$L$ edge cannot be counted twice.
Hence the number of rank-1 edges in $M$ is exactly equal to the value of the max-flow $f_1$ in $H_1$. 

It is straightforward to see that $M$ matches every applicant $a$ to a post in $f(a) \cup s(a)$, since these are the only edges 
added during the course of the algorithm. This completes the proof of the correctness of our algorithm.
\end{proof}

{\bf Running time: }The flow network has $O(|E|)$ vertices and edges. The maximum flow is at most $|A|$. So the running time of the algorithm is bounded by $O(|A||E|)$.

\section{Hardness for non-laminar classifications}
\label{sec:hardness}
In this section, we consider the \CRMM\ and \CPM\ problems where the classifications are not necessarily laminar.
We show that the following decision version of the \CRMM\ problem is \NP-hard: Given an instance $G=(A\cup P, E)$ of the \CRMM\ problem
and a signature vector $\sigma=(\sigma_1,\ldots,\sigma_r)$, 
does there exist a feasible matching $M$ in $G$ such that $M$ has a signature $\rho$ such that $\rho\succeq\sigma$? We give a reduction from the \monEsat\  problem to the above decision version of \CRMM. 
Throughout this section, we refer to this decision version as the \CRMM\ problem.
Our reduction also works for showing the hardness for \CPM\ problem, since only posts have classifications, and each applicant can be matched to at most one post.
Also, the reduction shows that the two problems remain \NP-hard for non-laminar classifications even when preference lists are strict and are of length two.

The \monEsat\ problem is a variant of the boolean satisfiability problem where the input is a conjunction of $m$ clauses.
Each clause is a disjunction of exactly three variables and no variable appears in negated form.
The goal is to decide whether there exists a truth assignment to the variables such that every clause has exactly one true variable
and hence two false variables. This problem is known to be NP-hard~\cite{Schaefer78}.
 Let $\phi$ 
be the given instance of the \monEsat\ problem, with $n$ variables $x_1,\ldots,x_n$
and $m$ clauses $C_1,  C_2, \ldots, C_m$.
We construct an instance $G = (A \cup P, E)$ of the \CRMM\ problem as follows:

\noindent {\bf Applicants: }For each variable $x_i$ in $\phi$, there are two applicants $a_i,b_i$ in $A$. For each occurrence of $x_i$ in clause $C_j$,
there are two applicants $a_{ij},b_{ij}$. Thus $A=\{a_i,b_i,a_{ij},b_{ij}\mid x_i \in \phi, x_i\in C_j\}$ and $|A| = 2n+6m$.  

\noindent {\bf Posts: }For each variable $x_i$, there are three posts $p_i, p^t_i$ and $p^f_i$. For each clause $C_j$, there is a post $p_j$. Thus $P=\{p_i, p^t_i,p^f_i\mid x_i\in \phi\}\cup \{p_j\mid C_j\in \phi\}$ and $|P| = 3n+m$.

\noindent {\bf Preferences of applicants: }The applicants have following preferences:

\begin{minipage}{0.5\linewidth}
\begin{eqnarray*}
a_i & : &\quad p_i, \quad p^t_i\\
b_i & : & \quad p_i, \quad p^f_i\\
\end{eqnarray*}
\end{minipage}
\begin{minipage}{0.4\linewidth}
\begin{eqnarray*}
a_{ij} & : & \quad p_j, \quad p^t_i\\
b_{ij} & : & \quad p_j, \quad p^f_i\\
\end{eqnarray*}
\end{minipage}

\noindent {\bf Quotas and classifications of posts: }
\begin{enumerate}
\item Let $C_j = x_i\vee x_{i'}\vee x_{i''}$;
the corresponding  post $p_j$ has quota $3$, and following classes:
\begin{enumerate}
\item $S_{ij}=\{a_{ij},b_{ij}\}$ with quota $1$ for each $x_i\in C_j$.
\item $S_{1j}=\{a_{ij},a_{i'j},a_{i''j}\}$ with quota $1$. 
\item $S_{2j}=\{b_{ij},b_{i'j},b_{i''j}\}$ with quota $2$.
\end{enumerate}
\item Each post $p^t_i$ has quota $k_i$=the number of occurrences of $x_i$ in $\phi$ and following classes:

$S^t_j=\{a_{ij},a_i\}$ with quota $1$, for each $j$ such that $x_i\in C_j$.
\item Each post $p^f_i$ has quota $k_i$=the number of occurrences of $x_i$ in $\phi$ and following classes:

$S^f_j=\{b_{ij},b_i\}$ with quota $1$, for each $j$ such that $x_i\in C_j$.
\item Each post $p_i$ has quota $1$ and no classes.
\end{enumerate}

We now show the correctness of the reduction for the \CRMM\ problem (Theorem \ref{thm:crmm-hard} stated in Section \ref{sec:intro}).
\begin{theorem}
The instance $G$ constructed above, corresponding to a given formula $\phi$ with $n$ variables and $m$ clauses, 
has a matching of signature $\sigma=(3m+n,3m+n)$ if and only if $\phi$ has a satisfying assignment.
\end{theorem}
\begin{proof}
Let $\phi$ have a satisfying assignment $\eta$. We show that $G$ has a matching $M$ with signature $\sigma$.
\begin{itemize}
\item If $x_i=1$ in $\eta$, set $M(a_{ij})=p_j$ and $M(b_{ij})=p^f_i$ for each clause $C_j$ containing $x_i$.
Set $M(a_i)=p^t_i$ and $M(b_i)=p_i$. 
\item If $x_i=0$ in $\eta$, set $M(b_{ij})=p_j$ and $M(a_{ij})=p^t_i$ for each clause $C_j$ containing $x_i$. 
Set $M(b_i)=p^f_i$ and $M(a_i)=p_i$. 
\end{itemize}
Note that $\eta$ satisfies the property that, in each clause $C_j$, $\eta$ assigns value $1$ to exactly one variable, 
say $x_i$, and value $0$ to remaining two variables $x_{i'}$ and $x_{i''}$. 
It is easy to see that $M$ satisfies all the quotas. Further, $M$ has signature $\sigma$, since each post $p_j$ gets matched to three applicants, each post $p_i$
is matched to one applicant, total number of applicants matched to posts $p^t_i$ and $p^f_i$ is exactly $\sum_{i=1}^n(k_i+1)=3m+n$.

Now consider a matching $M$ with signature $\sigma$ in $G$. We construct a satisfying assignment $\eta$ corresponding to $M$. There are $2n+6m$ applicants, 
so all the applicants must be matched in $M$ to achieve signature $\sigma$. Since the number of rank $1$
edges in $M$ is $3m+n$, $|M(p_j)|=3$ for each $j$. Let $C_j=x_i\vee x_{i'}\vee x_{i''}$. Due to the class $S_{1j}$ of $p_j$, at most one of $a_{ij}, a_{i'j}, a_{i''j}$
can be matched to $p_j$. Also, all three applicants matched to $p_j$ can not be $b_{ij}, b_{i'j}, b_{i''j}$ because of class $S_{2j}$. 
Therefore, exactly one of $a_{ij}, a_{i'j}, a_{i''j}$ must be in $M(p_j)$.
Without loss of generality, let $a_{ij}\in M(p_j)$. Then, due to class $S_{ij}$, $b_{ij}\notin M(p_j)$. Then $b_{i'j}, b_{i''j}\in M(p_j)$.
Also, $b_{ij}\in M(p^f_i)$, and due to class $S^f_j$, $b_i\notin M(p^f_i)$. Therefore $b_i\in M(p_i)$ and consequently, due to quota $1$ of $p_i$, $a_i\in M(p^t_i)$. 
This implies that, for each $C_j$ such that $x_i\in C_j$, $a_{ij}\notin M(p^t_i)$, due to the quota constraint of class $S^t_j$. Thus, if $a_{ij}\in M(p_j)$, $a_{ij'}\in M(p_{j'})$ 
for each clause $C_{j'}$ containing $x_i$. We set $x_i=1$ in $\eta$ in this case.

Now consider the case when $a_{ij}\notin M(p_j)$ for some $j$. Then $M(a_{ij})=p^t_i$ and hence $M(a_i)=p_i$, $M(b_i)=p^f_i$, forcing $M(b_{ij})=p_j$ for each $j$ such that $x_i$ appears in clause $C_j$ in $\phi$. Therefore, $M(a_{ij})=p^t_i$ for each $j$, where $x_i\in C_j$. We set $x_i=0$ in this case. It can be seen that exactly one variable
from each clause is set to $1$ in $\eta$ and hence $\eta$ is a satisfying assignment.
\end{proof}

We show correctness of the reduction for the \CPM\ problem using Theorem \ref{thm:pop-hard} below.

\begin{theorem}\label{thm:pop-hard}
The instance $G$ admits a popular matching amongst all feasible matchings if and only if the formula $\phi$ has a satisfying assignment.
\end{theorem}
\begin{proof}
From the characterization in Section~\ref{sec:pop}, a popular matching in $G$ must match all the applicants, all the $p_j$ and $p_i$
posts are $f$-posts whereas all the $p^t_i, p^f_i$ posts are $s$-posts, and any matching in $G$ that is a maximum feasible matching on rank-$1$ edges
and matches all the applicants is popular in $G$. Note that we do not need last resort posts here, since $s(a)\neq \emptyset$ for each $a\in A$.

The matching $M$ referred to in the proof of Theorem \ref{thm:crmm-hard} satisfies the above characterization, and hence is popular in $G$.
Thus the same proof as that of Theorem~\ref{thm:crmm-hard} works here.
\end{proof}

\noindent{\bf Remark:} The same instance $G$ without ranks on edges is useful in showing \NP-hardness of maximum cardinality feasible matching in the presence of classifications. 
This is because, in $G$, a maximum cardinality feasible matching has size $|A|$, that is, it matches all applicants, if and only if $\phi$ has a valid $1$-in-$3$ SAT assignment. 
Thus we establish Theorem~\ref{thm:max-hard} (Section~\ref{sec:intro}).

\bibliography{references}

\newpage

\end{document}